\documentclass{article}
\usepackage{style}

\newcommand*{\mypath}{figures/}
\usetikzlibrary{external}
\tikzexternaldisable

\title{Constrained school choice with incomplete information}
\author{Hugo Gimbert \and Claire Mathieu \and Simon Mauras}
\date{\today}

\begin{document}

\maketitle

\begin{abstract}
School choice is  the two-sided matching market where students (on one side) are to be matched with schools (on the other side) based on their mutual preferences. The classical algorithm to solve this problem is the celebrated deferred acceptance procedure, proposed by Gale and Shapley.
After both sides have revealed their mutual preferences, the algorithm computes an optimal stable matching.
Most often in practice, notably when the process is implemented by a national clearinghouse and thousands of schools enter the market, there is a quota on the number of applications that a student can submit: students have to perform  a partial revelation of their preferences, based on partial information on the market.
We model this situation by drawing each student type from a publicly known distribution and study Nash equilibria of the corresponding Bayesian game.
We focus on symmetric equilibria, in which all students play the same strategy.
We show existence of these equilibria in the general case, and provide two algorithms to compute such equilibria under additional assumptions, including the case where schools have identical preferences over students.
\end{abstract}

\section{Introduction}
\label{sec:constrained:intro}

School choice is referred in the literature as the two-sided matching market where students (on one side) are to be matched with schools (on the other side) based on their mutual preferences. The classical algorithm to solve this problem is the celebrated deferred acceptance procedure, proposed by Gale and Shapley~\cite{gale1962college}, and since implemented by many clearinghouse \cite{abdulkadirouglu2003school,abdulkadirouglu2005new,correa2019school}.
Most often in practice, the clearinghouse sets an upper quota on the number of applications each student can submit. This requires a strategic behaviour from students who should find a balance between applications to top-tier schools and applications to less attractive but also less selective lower-tier schools.

The existence and computability of Nash equilibria is a desirable property for two reasons.
First, Nash equilibria are among the possible long-term outcomes of the market, possibly emerging after a series of best-response dynamics or evolutional selection of strategies. Second, and most importantly, being able to compute a Nash equilibrium  provides a solid basis to develop a recommendation system in order to help the students to select the schools they want to apply~to.

In case complete information about the preferences of the students and schools is available,
implementing a Nash equilibrium is rather easy.
The strategic behaviour of students limited to a fixed number of options was studied by Romero-Medina \cite{romero1998implementation}, and later investigated by Calsamiglia, Haeringer and Klijn \cite{haeringer2009constrained,calsamiglia2010constrained}.
After a pre-computation of the student-optimal stable matching $\mu_{\mathcal S}$, a simple recommendation can be made to every student matched in $\mu_{\mathcal S}$: they only need to apply to a single school, their match in $\mu_{\mathcal S}$. 
As a direct corollary of~\cite{dubins1981machiavelli,roth1982economics}, this leads to a Nash equilibrium. Remark that student unmatched in $\mu_{\mathcal S}$ wont be matched in this equilibrium, whatever strategy they choose.

In practice, assuming that the student-optimal stable matching is computable \emph{ex-ante} is rather unrealistic. For example, in the French college admission system ``Parcoursup", there are more than 900000 students signed up.
Applicants should report a shortlist of 20 wishes before a fixed deadline. Based on statistics of previous years, one might evaluate how the grades of a particular student compare to others and evaluate the percentage of students that will be ranked higher in a particular school. But acquiring before the deadline the information needed for an exact computation of the student-optimal matching is unfeasible.

In this paper, we propose a formal model for the constrained school choice with incomplete information, and study the existence and computability of Nash equilibria in the associated 
incomplete information game.
%omplete information and are able to compute the
%students only have a distributional knowledge on the preferences of schools and other students:
In our model, each student draws a type from a publicly known distribution $\mu$ (see \cref{sec:constrained:model}). In \cref{sec:constrained:examples}, we detail interesting examples that can be used to state recommendations for students, schools and decision makers. In \cref{sec:constrained:existence}, we give the  proof of existence of a symmetric Bayes-Nash equilibrium (\cref{thm:existence}). In \cref{sec:constrained:efficient} we give efficient algorithms to compute equilibria when the number of types is finite and additional hypotheses are made, including the case where schools have identical preferences over students (\cref{thm:computealpha,thm:compute-identical,thm:compute-strongalpha}). In \cref{sec:constrained:approx} we prove a convergence theorem, showing that one can compute an equilibrium for a game with a continuous type distribution $\mu$, using a (weakly) converging sequence of distributions $(\mu_k)_{k\geq 1}$ having finite supports (\cref{thm:convergence}).

\vspace{-.4cm}
\paragraph{Related work.}
This paper is closely related to the literature of matching under random preferences. Pittel \cite{pittel1992likely} study balanced matching markets with uniformly random preferences. Rephrasing his results in our setting, the student-proposing deferred acceptance procedure matches with high probability every student to one of her top $\log^2 n$ choices, which proves that an upper quota of $\log^2 n$ applications per student does not deteriorate the outcome.
Immorlica and Mahdian \cite{immorlica2015incentives}, and Kojima and Pathak \cite{kojima2009incentives} study matching markets where one side of the market has random preference lists of constant size, and show that such markets have a (nearly) unique stable matching. When quotas are constant and preferences are uniform, this implies that games based on student or school proposing deferred acceptance are (almost) the same.

More recent papers discuss the effect of an upper quota on the number of applications. Beyhaghi, Saban and Tardos \cite{beyhaghi2017effect} study the efficiency of equilibria in a model where each side is divided into two uniform tiers, and each student chooses her number of applications to top-tier schools.
Beyhaghi and Tardos~\cite{beyhaghi2021randomness} study the social welfare (size of the matching) as a function of the number of applications, in a model where preferences of agents are drawn uniformly at random. 
Echenique, Gonzalez, Wilson and Yariv~\cite{echenique2020top} examine the National Resident Matching Program and argue that doctors are strategic when reporting their preferences.

The best response of a student to the strategies of others is related to the simulatenous search literature. Chade and Smith~\cite{chade2006simultaneous} discuss the problem where one student must choose a portfolio of schools in which she applies: each application has a cost, a probability of success and a cardinal utility when successful. Ali and Shorrer~\cite{ali2021college} generalize their model to allow correlations between admission decisions.

\vspace{-.4cm}
\paragraph{Takeaway message.}
In general, the deferred acceptance mechanism is known to be strategy-proof for the proposing side \cite{dubins1981machiavelli}, but no mechanism is truthful for both sides of the market \cite{roth1982economics}. However, empirical results show that the stable matching is often unique \cite{roth1999redesign}, in which case stable matching procedures are truthful for all agents, even when they have incomplete information~\cite{ehlers2007incomplete}.
Thus, having a unique stable matching is a desirable property, and we argue that this fact carries over to the case where students have restricted preferences. First, in terms of number of equilibria, examples (see \cref{sec:constrained:examples}) illustrate that the fewer stable matchings there are, the fewer equilibria the game has. Second, in terms of outcome, multiple stable matchings can induce outcomes which are unstable (see \cref{sec:constrained:complete}) or sub-optimal (see \cref{sec:constrained:reversed}). And finally, in terms of computability of an equilibrium, \cref{sec:constrained:efficient} give two algorithms to compute equilibria, under extra hypotheses borrowed from the literature of unique stable matchings.

\vspace{-.2cm}
\section{The model}
\label{sec:constrained:model}

We consider a game where players are students who do not know the exact preferences of other students.
For the sake of modeling, each student has a type $T = [0,1]^d$ with $d \geq 1$, which can be thought as a feature vector representing both her preferences and characteristics. Types are drawn without replacement\footnote{Types are drawn without replacement in order to have a well defined game when the distribution $\mu$ is discrete. This does not mean students cannot have the same preferences over schools, as one can duplicate types by increading the dimension $d$ of the type space. When the distribution is non-atomic,  types are drawn independently.} from the set of types, using a probability distribution $\mu \in \Delta(T)$.
Each student~$i$ knows her own type $t_i\in T$ (private information) and the distribution $\mu$ (common information).
Each school $j$ has a capacity $c_j$, a bounded measurable value function $v_j: T\rightarrow\mathbb R_+$ and a measurable scoring function $s_j: T\rightarrow[0,1]$.

The set of actions $A$ is the set of preference lists containing at most $\ell$ schools\footnote{More generally, results of this paper hold if $A$ is an arbitrary subset of preference lists over schools.}.
Each student~$i$ reports a preference list $a_i\in A$.
Schools sort students by decreasing score, breaking ties uniformly at random. Then, we compute a matching using the student proposing deferred acceptance algorithm. Each student $i$ receives a utility $v_j(t_i)$ if she is assigned to school $j$, and a utility of $0$ is she stays unmatched.

\begin{algorithm}
    \begin{algorithmic}
    \State Game parameters: $n$, $m$, $A$, $T$, $(v_j)_{j\in[m]}$, $(c_j)_{j\in[m]}$ and $(s_j)_{j\in [m]}$.
        \Function{Utility}{$(t_i, a_i)_{i\in[n]}\in (T\times A)^n$}
            \State Student $i$ has type $t_i$, reports the preference list $a_i$, and her score at school $j$ is $s_j(t_i)$.
            \State School $j$ has capacity $c_j$, sorts students by decreasing scores,
            \\\hspace{1cm} breaking ties uniformly at random.
            \State Students are assigned to schools using the student proposing deferred acceptance algorithm.
            \State Each student $i$ receives utility $u_i = v_j(t_i)$ if she assigned to school $j$, 
            \\\hspace{1cm} and $u_i= 0$ if she is unassigned. 
            \State\textbf{Return} the vector of utilities $(u_i)_{i\in[n]}$, averaged over all possible tie-breaking choices.
        \EndFunction
    \end{algorithmic}
    \caption{Description of the matching game}
    \label{algo:game-bayesian}
\end{algorithm}

Students choose their actions strategically:
the set of (behavioral) strategies $\mathcal S$ is the set of measurable function $p: T\rightarrow \Delta(A)$.
A strategy profile is a vector of strategies $(p_i)_{i\in[n]} \in \mathcal S^n$, where $p_i$ is the strategy of student $i$ and $p_{-i}$ denotes the vector of strategies of all students except $i$.
Under this strategy profile, the expected payoff of student $i$ is denoted
$U_\mu(p_i, p_{-i})$, this is the $i$-th component of the vector
$\mathbb E_{t,a}[\textsc{Utility}((t_i, a_i)_{i\in[n]})]$, where the expectation is taken over the random draws of $t$ and $a$: $t_i$'s are drawn without replacement from $\mu$ and each $a_i$ is drawn from $p_i(t_i)$.
Notice that this definition already incorporates the symmetry of the game:
if the strategies in $p_{-i}$ are permuted, the payoff of player  $i$ does not change, and if we swap two players, their payoffs are swapped accordingly. In other words, the payoff of a player only depends on his own strategy and the multiset of strategies played by the other players, independently of each player's identity.

A strategy profile is a Bayes-Nash equilibrium if each student cannot improve her utility by deviating from the strategy profile. More precisely, $(p_i)_{i\in[n]} \in \mathcal S^n$ is an equilibrium if $U_\mu(p^*,p_{-i}) \leq U_\mu(p_i,p_{-i})$ for every $i\in[n]$ and $p^* \in \mathcal S$.

Our first theorem states that the strategic part of the game for a student is to choose her (unordered) set of applications. More precisely, once she decided which schools she will apply to, it is optimal for her to sort schools by decreasing value. As a Corollary, when $\ell = m$, the set of actions $A$ is unconstrained and contains all the permutations over schools, thus sorting schools by decreasing score is a dominant strategy.

\begin{theorem}
\label{thm:sorted}
    Let $t\in T$ be a type and $a\in A$ be an action, and define $a^*$ the preference list where schools from $a$ are sorted by non-increasing order of value $v_j(t)$. If $a^*$ is a valid action, then for a student of type $t$ reporting $a^*$ dominates reporting $a$. 
\end{theorem}
\begin{proof}
  See Theorem 9 from \cite{dubins1981machiavelli} or Lemma 8.1 from \cite{haeringer2009constrained}.
\end{proof}

\section{Motivating examples}
\label{sec:constrained:examples}

\subsection{Complete information}
\label{sec:constrained:complete}

Recall that types of students are drawn without replacement from $\mu$. Thus, if $\mu$ is a discrete distribution with a finite support of size $n$, then students exactly know the types of other students, which proves that complete information is a special case of our model.

Haeringer and Klijn \cite{haeringer2009constrained} study an equivalent complete information game: $n$ students and $m$ schools have ordinal preferences over one another, and each student must report preference lists of length at most $\ell$ to the clearinghouse. When $\ell=1$, they show that each stable matching can be implemented at equilibrium (Proposition 6.1), and the outcome of every equilibrium is a stable matching (Proposition 6.3). Additionally\footnote{When $\ell \geq 1$, Haeringer and Klijn \cite{haeringer2009constrained} also give (Theorem 6.6) a necessary and sufficient condition on the preferences of schools such that the outcome is stable for every preferences of students and for every equilibrium. This result is in general incomparable with our \cref{thm:unique-equilibrium}}, they give examples to show that when $\ell > 1$, the outcome of some equilibrium can be unstable (Examples 6.6, 8.3, 8.4 and 8.5). In \cref{fig:equilibrium-unstable}, we reproduce Example 8.3 from \cite{haeringer2009constrained}.

%%%%%%%%%%%%%%%%%%%%%%%%%%%%%%%%%%%%%%%%%%%%%%%%%%%%%%%%%%%%%%%%%%%%%%%%
%%%%%%%%%%%%%%%%%%%%%%%%%%%%%%%%%%%%%%%%%%%%%%%%%%%%%%%%%%%%%%%%%%%%%%%%
\begin{myfigure}[h!]
\centering
\begin{tikzpicture}[scale=2,line width=1,>=stealth]
    \stickfigure{0}{1.5} \node at (0,1.5) {$a$};
    \stickfigure{0}{1.0} \node at (0,1.0) {$b$};
    \stickfigure{0}{0.5} \node at (0,0.5) {$c$};
    \stickfigure{0}{0.0} \node at (0,0.0) {$d$};
    \stickhouse{1}{1.25} \node at (1,1.25) {$1$};
    \stickhouse{1}{0.75} \node at (1,0.75) {$2$};
    \stickhouse{1}{0.25} \node at (1,0.25) {$3$};
    \node[anchor=east] at (-0.2,1.5) {$1 \succ 2 \succ 3$};
    \node[anchor=east] at (-0.2,1.0) {$2 \succ 3 \succ 1$};
    \node[anchor=east] at (-0.2,0.5) {$3 \succ 1 \succ 2$};
    \node[anchor=east] at (-0.2,0.0) {$1 \succ 2 \succ 3$};
    \node[anchor=west] at (1.2,1.25) {$c \succ a \succ b \succ d$};
    \node[anchor=west] at (1.2,0.75) {$a \succ b \succ c \succ d$};
    \node[anchor=west] at (1.2,0.25) {$b \succ d \succ c \succ a$};
    \draw (0.2,1.5) -- (0.8,0.75);
    \draw (0.2,1.0) -- (0.8,0.25);
    \draw (0.2,0.5) -- (0.8,1.25);
\end{tikzpicture}
\hfill
\begin{tikzpicture}[scale=2,line width=1,>=stealth]
    \stickfigure{0}{1.5} \node at (0,1.5) {$a$};
    \stickfigure{0}{1.0} \node at (0,1.0) {$b$};
    \stickfigure{0}{0.5} \node at (0,0.5) {$c$};
    \stickfigure{0}{0.0} \node at (0,0.0) {$d$};
    \stickhouse{1}{1.25} \node at (1,1.25) {$1$};
    \stickhouse{1}{0.75} \node at (1,0.75) {$2$};
    \stickhouse{1}{0.25} \node at (1,0.25) {$3$};
    \node[anchor=east] at (-0.2,1.5) {$1 \succ 2$};
    \node[anchor=east] at (-0.2,1.0) {$2 \succ 3$};
    \node[anchor=east] at (-0.2,0.5) {$3 \succ 1$};
    \node[anchor=east] at (-0.2,0.0) {$1 \succ 2$};
    \node[anchor=west] at (1.2,1.25) {$c \succ a \succ b \succ d$};
    \node[anchor=west] at (1.2,0.75) {$a \succ b \succ c \succ d$};
    \node[anchor=west] at (1.2,0.25) {$b \succ d \succ c \succ a$};
    \draw (0.2,1.5) -- (0.8,1.25);
    \draw (0.2,1.0) -- (0.8,0.75);
    \draw (0.2,0.5) -- (0.8,0.25);
\end{tikzpicture}
\caption{The true preferences of students are displayed on the left: there is a unique stable matching where students $a$, $b$ and $c$ get their 2\textsuperscript{nd} choices. When students are restricted to apply to at most two schools, truncating each student preference list yields an equilibrium: students $a$, $b$ and $c$ get their 1\textsuperscript{st} choice, whereas $d$ will stay unmatched no matter what she reports. The outcome is not stable as student $d$ and school $3$ prefer each other to their respective partners. Observe that if $d$ applies to school $3$ then students $a$, $b$ and $c$ will get their 2\textsuperscript{nd} choices, which is a different equilibrium where the outcome is stable.}
\label{fig:equilibrium-unstable}
\end{myfigure}
%%%%%%%%%%%%%%%%%%%%%%%%%%%%%%%%%%%%%%%%%%%%%%%%%%%%%%%%%%%%%%%%%%%%%%%%%%
%%%%%%%%%%%%%%%%%%%%%%%%%%%%%%%%%%%%%%%%%%%%%%%%%%%%%%%%%%%%%%%%%%%%%%%%%%

\vspace{-.5cm}
In general, because every stable matching can be implemented by an equilibrium, a necessary condition for the outcome to be unique is to have a unique stable matching. This condition however is not sufficient, as \cref{fig:equilibrium-unstable} illustrates with an example having a unique stable matching but several possible outcomes when $\ell=2$. In \cref{thm:unique-equilibrium}, we show that $\alpha$-reducibility is a sufficient condition for having a unique outcome. The notion of $\alpha$-reducibility was introduced by Alcade \cite{alcalde1994exchange} in the context of stable roommates, then investigated by Clark \cite{clark2006uniqueness} who showed it is equivalent to having a unique stable matching in every sub-market (Theorems 4 and 5).

\begin{definition}[$\alpha$-reducibility]
\label{def:alpha}
    We say that a two-sided matching market is $\alpha$-reducible if for every subset of students $A \subseteq [n]$ and subset of schools $B \subseteq [m]$, there exist a fixed pair $(i,j) \in A\times B$ such that $i$ and $j$ prefer each other to everyone else in $A$ and $B$.
\end{definition}

\begin{theorem}\label{thm:unique-equilibrium}
    If the two-sided matching market is $\alpha$-reducible, then for every $\ell \geq 1$ the outcome of every Nash equilibrium is the unique stable matching.
\end{theorem}
\begin{proof}
    For every Nash equilibrium, start the analysis by setting $A = [n]$ and $B = [m]$. From $\alpha$-reducibility we know that there is a fixed pair $(i,j) \in A\times B$. Student $i$ can ensure she is matched with her first choice $j$, thus this must be her outcome by definition of a Nash equilibrium. We remove $i$ from $A$, decrease the capacity of $j$ and remove it from $B$ if it reached $0$. We continue with the same reasoning by induction.
\end{proof}

\subsection{Reversed preferences}
\label{sec:constrained:reversed}

Consider a simple example with $n=2$ students and $m=2$ schools of capacity $c_1=c_2=1$. The set of types $T = [0,1]^2$ is two-dimensional. A student of type $(x,y)$ gives the value $v_1(x,y) = r+x$ to school 1, and the value $v_2(x,y) = r+y$ to school 2, where $r$ is a positive constant representing how risk-averse the students are. The preferences of schools are reversed, in the sense that a student of type $(x,y)$ has a score of $s_1(x,y) = y$ at school 1 and a score of $s_2(x,y) = x$ at school~2.
\Cref{fig:example-1.1} illustrates a situation with two stable matchings. This occurs with probability $1/2$ when $\mu$ is uniform over the diagonal $\{(x,y)\in T\,|\,x+y=1\}$, and with probability $1/6$ when the distribution $\mu$ is uniform over $T$.  

%%%%%%%%%%%%%%%%%%%%%%%%%%%%%%%%%%%%%%%%%%%%%%%%%%%%%%%%%%%%%%%%%%%%%%%%%%
%%%%%%%%%%%%%%%%%%%%%%%%%%%%%%%%%%%%%%%%%%%%%%%%%%%%%%%%%%%%%%%%%%%%%%%%%%
\begin{myfigure}
    \centering
    \begin{tikzpicture}[scale=3,line width=1,>=stealth]
        %%%%%
        \stickfigure{2}{0.8}\node at(2,0.8) {$a$};
        \stickfigure{2}{0.2}\node at(2,0.2) {$b$};
        \draw[->] (2.2,0.7) -- (2.8,0.3);
        \draw[->] (2.2,0.3) -- (2.8,0.7);
        \draw[->] (2.8,0.8) -- (2.2,0.8);
        \draw[->] (2.8,0.2) -- (2.2,0.2);
        %%%%%
        \begin{scope}[color=green!50!black]
            \fill[opacity=0.2] (0,0) -- (1,1) -- (1,0);
            \draw[dashed] (0,0.8) -- (1,0.8);
            \draw[dashed] (0,0.2) -- (1,0.2);
            \node[anchor=south,rotate=-90] at (1,0.5) {school 1};
            \stickhouse{3}{0.8}\node at(3,0.8) {1};
        \end{scope}
        %%%%%
        \begin{scope}[color=blue!50!black]
            \draw[dashed] (0.3,0) -- (0.3,1);
            \draw[dashed] (0.6,0) -- (0.6,1);
            \fill[opacity=0.2] (0,0) -- (1,1) -- (0,1);
            \node[anchor=south] at (0.5,1) {school 2};
            \stickhouse{3}{0.2}\node at(3,0.2) {2};
        \end{scope}
        %%%%%
        \draw[->] (0,0) -- (1.1,0);
        \draw[->] (0,0) -- (0,1.1);
        \draw (1,0) -- (1,1) -- (0,1);
        %\draw (0,0) -- (1,1);
        \node at (1.2,0) {$x$};
        \node at (0,1.2) {$y$};
        \fill (0.3,0.8) circle (0.03);
        \fill (0.6,0.2) circle (0.03);
        \node[anchor=north] at (0.3,0) {$x_a$};
        \node[anchor=north] at (0.6,0) {$x_b$};
        \node[anchor=east] at (0,0.8) {$y_a$};
        \node[anchor=east] at (0,0.2) {$y_b$};
    \end{tikzpicture}
    \caption{An example with two stable matchings: student~$a$ prefers school~2 (upper triangle $x_a < y_a$), student~$b$ prefers school 1 (lower triangle $y_b < x_b$), school~1 prefers student~$a$ (horizontal lines $y_b < y_a$), and school 2 prefers student $b$ (vertical lines $x_a < x_b$).}
    \label{fig:example-1.1}
\end{myfigure}

%%%%%%%%%%%%%%%%%%%%%%%%%%%%%%%%%%%%%%%%%%%%%%%%%%%%%%%%%%%%%%%%%%%%%%%%%%
%%%%%%%%%%%%%%%%%%%%%%%%%%%%%%%%%%%%%%%%%%%%%%%%%%%%%%%%%%%%%%%%%%%%%%%%%%

\begin{myfigure}[p!]
    \begin{center}
    \begin{tikzpicture}[xscale=2.5,yscale=1.5,line width=1,>=stealth]
        \draw[->] (0,0) -- (1.1,0);
        \draw[->] (0,0) -- (0,1.6);
        \draw (0.5,-.05) -- (0.5,.05);
        \draw (1,-.05) -- (1,.05);
        \draw (-.03,1) -- (.03,1);
        \draw (-.03,0.7) -- (.03,0.7);
        \draw[dotted] (0,1) -- (1,1);
        \draw[dotted] (0,0.7) -- (1,0.7);
        \draw[domain=0:0.5,smooth, variable=\x, green!50!black]
            plot ({\x}, {(0.7+\x)});
        \draw[domain=0.5:1,smooth, variable=\x, green!50!black]
            plot ({\x}, {(0.7+\x)*(0.5+1-\x)});
        \draw[domain=0:0.5,smooth, variable=\x, blue!50!black]
            plot ({\x}, {(0.7+1-\x)*(0.5+\x)});
        \draw[domain=0.5:1,smooth, variable=\x, blue!50!black]
            plot ({\x}, {(0.7+1-\x)});
        \node[anchor=east] at (-.02,1) {$1$};
        \node[anchor=east] at (-.02,0.7) {$r$};
        \node[anchor=west] at (1.1,0) {$x$};
        \fill[blue!50!black,opacity=0.5] (0,-.05) rectangle (0.5,.05);
        \fill[green!50!black,opacity=0.5] (0.5,-.05) rectangle (1,.05);
        \node at (.5,1.5) {$r < 1$};
    \end{tikzpicture}
    \begin{tikzpicture}[xscale=2.5,yscale=1.5,line width=1,>=stealth]
        \draw[->] (0,0) -- (1.1,0);
        \draw[->] (0,0) -- (0,1.6);
        \draw (1,-.05) -- (1,.05);
        \draw (-.03,1) -- (.03,1);
        \draw[dotted] (0,1) -- (1,1);
        \draw[domain=0:1,smooth, variable=\x, cyan!50!black]
            plot ({\x}, {(1+\x)*(1-\x/2)});
        \node[anchor=east] at (-.02,1) {$1$};
        \node[anchor=west] at (1.1,0) {$x$};
        \fill[cyan!50!black,opacity=0.5] (0,-.05) rectangle (1,.05);
        \node at (.5,1.5) {$r = 1$};
    \end{tikzpicture}
    \begin{tikzpicture}[xscale=2.5,yscale=1.5,line width=1,>=stealth]
        \draw[->] (0,0) -- (1.1,0);
        \draw[->] (0,0) -- (0,1.6);
        \draw (0.5,-.05) -- (0.5,.05);
        \draw (1,-.05) -- (1,.05);
        \draw (-.03,1) -- (.03,1);
        \draw (-.03,1.3) -- (.03,1.3);
        \draw[dotted] (0,1) -- (1,1);
        \draw[dotted] (0,1.3) -- (1,1.3);
        \draw[domain=0:0.5,smooth, variable=\x, green!50!black]
            plot ({\x}, {(1.3+\x)*(1-\x)});
        \draw[domain=0.5:1,smooth, variable=\x, green!50!black]
            plot ({\x}, {(1.3+\x)/2});
        \draw[domain=0:0.5,smooth, variable=\x, blue!50!black]
            plot ({\x}, {(1.3+1-\x)/2});
        \draw[domain=0.5:1,smooth, variable=\x, blue!50!black]
            plot ({\x}, {(1.3+1-\x)*(\x)});
        \node[anchor=east] at (-.02,1) {$1$};
        \node[anchor=east] at (-.02,1.3) {$r$};
        \node[anchor=west] at (1.1,0) {$x$};
        \fill[green!50!black,opacity=0.5] (0,-.05) rectangle (0.5,.05);
        \fill[blue!50!black,opacity=0.5] (0.5,-.05) rectangle (1,.05);
        \node at (.5,1.5) {$1 < r$};
    \end{tikzpicture}
    \qquad
    \begin{tikzpicture}[line width=1,scale=2]
        \useasboundingbox (0,-.5) rectangle (0.8,0.7);
        \fill[green!50!black,opacity=0.5] (0,0.4) rectangle (0.2,0.5);
        \fill[cyan!50!black,opacity=0.5] (0,0.2) rectangle (0.2,0.3);
        \fill[blue!50!black,opacity=0.5] (0,0) rectangle (0.2,0.1);
        \draw (0,0.4) rectangle (0.2,0.5);
        \draw (0,0.2) rectangle (0.2,0.3);
        \draw (0,0) rectangle (0.2,0.1);
        \node[anchor=west] at (0.2,0.45) {school 1};
        \node[anchor=west] at (0.2,0.25) {mixed};
        \node[anchor=west] at (0.2,0.05) {school 2};
    \end{tikzpicture}
    \end{center}
    \textbf{(a)} Family of equilibria when the distribution $\mu$ is uniform over the diagonal $\{(x,y)\in T\,|\,x+y=1\}$. The strategy of a student having type $(x,1-x)$ is represented by the color of the point on the $x$-axis. The different plots corresponds to the expected utility from each action, for a student of type $(x,1-x)$, when the other students play the equilibrium strategy.
    \begin{center}
    \begin{tikzpicture}[scale=2,line width=1,>=stealth]
        \useasboundingbox (-.2,-.1) rectangle (1.2,1.2);
        \fill[green!50!black,opacity=0.5] (0,0) -- (1,1) -- (1,0);
        \fill[blue!50!black,opacity=0.5] (0,0) -- (1,1) -- (0,1);
        \draw[->] (0,0) -- (1.1,0);
        \draw[->] (0,0) -- (0,1.1);
        \draw (0,0) -- (1,1);
        \draw (1,0) -- (1,1) -- (0,1);
        \node at (.5,1.1) {$r \leq 1/2$};
    \end{tikzpicture}
    \begin{tikzpicture}[scale=2,line width=1,>=stealth]
        \useasboundingbox (-.2,-.1) rectangle (1.2,1.2);
        \fill[cyan!50!black,opacity=0.5] (0,0) rectangle (.4,.4);
        \fill[blue!50!black,opacity=0.5] (0,.4) -- (.4,.4) -- (1,1) -- (0,1);
        \fill[green!50!black,opacity=0.5] (.4,0) -- (.4,.4) -- (1,1) -- (1,0);
        \draw[->] (0,0) -- (1.1,0);
        \draw[->] (0,0) -- (0,1.1);
        \draw (1,0) -- (1,1) -- (0,1);
        \draw (.4,0) -- (.4,.4) -- (0,.4);
        \draw (.4,.4) -- (1,1);
        \node[anchor=north] at (.4,0) {$2-1/r$};
        %\node[anchor=east] at (0,.4) {$t$};
        \node at (.5,1.1) {$1/2 \leq r \leq 1$};
    \end{tikzpicture}
    \begin{tikzpicture}[scale=2,line width=1,>=stealth]
        \useasboundingbox (-.2,-.1) rectangle (1.2,1.2);
        \fill[cyan!50!black,opacity=0.5] (0,0) rectangle (1,1);
        \draw[->] (0,0) -- (1.1,0);
        \draw[->] (0,0) -- (0,1.1);
        \draw (1,0) -- (1,1) -- (0,1);
        \node at (.5,1.1) {$r = 1$};
    \end{tikzpicture}
    \begin{tikzpicture}[scale=2,line width=1,>=stealth]
        \useasboundingbox (-.2,-.1) rectangle (1.2,1.2);
        \fill[cyan!50!black,opacity=0.5] (0,0) rectangle (.4,.4);
        \fill[green!50!black,opacity=0.5] (0,.4) -- (.4,.4) -- (1,1) -- (0,1);
        \fill[blue!50!black,opacity=0.5] (.4,0) -- (.4,.4) -- (1,1) -- (1,0);
        \draw[->] (0,0) -- (1.1,0);
        \draw[->] (0,0) -- (0,1.1);
        \draw (1,0) -- (1,1) -- (0,1);
        \draw (.4,0) -- (.4,.4) -- (0,.4);
        \draw (.4,.4) -- (1,1);
        \node[anchor=north] at (.4,0) {$1/r$};
        %\node[anchor=east] at (0,.4) {$t$};
        \node at (.5,1.1) {$1 \leq r$};
    \end{tikzpicture}
    \begin{tikzpicture}[line width=1,scale=2]
        \useasboundingbox (0,-.5) rectangle (0.8,0.7);
        \fill[green!50!black,opacity=0.5] (0,0.4) rectangle (0.2,0.5);
        \fill[cyan!50!black,opacity=0.5] (0,0.2) rectangle (0.2,0.3);
        \fill[blue!50!black,opacity=0.5] (0,0) rectangle (0.2,0.1);
        \draw (0,0.4) rectangle (0.2,0.5);
        \draw (0,0.2) rectangle (0.2,0.3);
        \draw (0,0) rectangle (0.2,0.1);
        \node[anchor=west] at (0.2,0.45) {school 1};
        \node[anchor=west] at (0.2,0.25) {mixed};
        \node[anchor=west] at (0.2,0.05) {school 2};
    \end{tikzpicture}
    \end{center}
    \textbf{(b)} Family of equilibria when the distribution $\mu$ is uniform over $T$. The strategy of a student having type $(x,y)$ is represented by the color of the point at those coordinates. 
    \caption{Family of symmetric Bayes-Nash equilibria when $\ell=1$ in the game described in \cref{fig:example-1.1}. When students are not risk averse ($r < 1$), they tend to apply according to their preferences. When students are risk averse ($r > 1$), they tend to apply according to their chance of being accepted.}
    \label{fig:example-1.2}
\end{myfigure}

%%%%%%%%%%%%%%%%%%%%%%%%%%%%%%%%%%%%%%%%%%%%%%%%%%%%%%%%%%%%%%%%%%%%%%%%%%
%%%%%%%%%%%%%%%%%%%%%%%%%%%%%%%%%%%%%%%%%%%%%%%%%%%%%%%%%%%%%%%%%%%%%%%%%%

\begin{myfigure}[p!]
    \begin{center}
    \begin{tikzpicture}[xscale=2,yscale=2.5,line width=1,>=stealth]
        \draw (0,0.6) -- (0,0.8);
        \draw (-.05,0.78) -- (0.05, 0.82);
        \draw (-.05,0.82) -- (0.05, 0.86);
        \draw[->] (0,0.84) -- (0,1.6);
        \begin{scope}[yshift=0.6cm]
            \draw[->] (0,0) -- (1.1,0);
            \draw (0.5,-.05) -- (0.5,.05);
            \draw (1,-.05) -- (1,.05);
            \fill[blue!50!black,opacity=0.5] (0,-.05) rectangle (0.5,.05);
            \fill[green!50!black,opacity=0.5] (0.5,-.05) rectangle (1,.05);
        \end{scope}
        \draw[dotted] (0,1) -- (1,1);
        \draw[domain=0:0.5,smooth, variable=\x, green!50!black]
            plot ({\x}, {(1+\x)});
        \draw[domain=0.5:1,smooth, variable=\x, green!50!black]
            plot ({\x}, {(1+\x)*(0.5+1-\x)});
        \draw[domain=0:0.5,smooth, variable=\x, blue!50!black]
            plot ({\x}, {(1+1-\x)*(0.5+\x)});
        \draw[domain=0.5:1,smooth, variable=\x, blue!50!black]
            plot ({\x}, {(1+1-\x)});
        \node[anchor=east] at (-.0,1) {$1$};
    \end{tikzpicture}
    \begin{tikzpicture}[xscale=2,yscale=2.5,line width=1,>=stealth]
        \def\t{0.2}
        \draw (0,0.6) -- (0,0.8);
        \draw (-.05,0.78) -- (0.05, 0.82);
        \draw (-.05,0.82) -- (0.05, 0.86);
        \draw[->] (0,0.84) -- (0,1.6);
        \begin{scope}[yshift=0.6cm]
            \draw[->] (0,0) -- (1.1,0);
            \fill[cyan!50!black,opacity=0.5] (0,-.05) rectangle (\t,.05);
            \fill[blue!50!black,opacity=0.5] (\t,-.05) rectangle (0.5,.05);
            \fill[green!50!black,opacity=0.5] (0.5,-.05) rectangle (1-\t,.05);
            \fill[cyan!50!black,opacity=0.5] (1-\t,-.05) rectangle (1,.05);
            \draw (\t,-.05) -- (\t,.05);
            \draw (0.5,-.05) -- (0.5,.05);
            \draw (1-\t,-.05) -- (1-\t,.05);
            \draw (1,-.05) -- (1,.05);
        \end{scope}
        \draw[dotted] (0,1) -- (1,1);
        \draw[domain=\t:0.5,smooth, variable=\x, green!50!black]
            plot ({\x}, {(1+\x)*(1-\t/2)});
        \draw[domain=\t:0.5,smooth, variable=\x, blue!50!black]
            plot ({\x}, {(1+1-\x)*(1/2-\t/2+\x)});
        \draw[domain=0.5:1-\t,smooth, variable=\x, green!50!black]
            plot ({\x}, {(1+\x)*(3/2-\t/2-\x)});
        \draw[domain=0.5:1-\t,smooth, variable=\x, blue!50!black]
            plot ({\x}, {(1+1-\x)*(1-\t/2)});
        \draw[domain=0:\t,smooth, variable=\x, cyan!50!black]
            plot ({\x}, {(1+\x)*(1-\x/2)});
        \draw[domain=1-\t:1,smooth, variable=\x, cyan!50!black]
            plot ({\x}, {(1+\x)*(1-\x/2)});
        \node[anchor=east] at (-.0,1) {$1$};
    \end{tikzpicture}
    \begin{tikzpicture}[xscale=2,yscale=2.5,line width=1,>=stealth]
        \draw (0,0.6) -- (0,0.8);
        \draw (-.05,0.78) -- (0.05, 0.82);
        \draw (-.05,0.82) -- (0.05, 0.86);
        \draw[->] (0,0.84) -- (0,1.6);
        \begin{scope}[yshift=0.6cm]
            \draw[->] (0,0) -- (1.1,0);
            \fill[cyan!50!black,opacity=0.5] (0,-.05) rectangle (1,.05);
            \draw (1,-.05) -- (1,.05);
        \end{scope}
        \draw[dotted] (0,1) -- (1,1);
        \draw[domain=0:1,smooth, variable=\x, cyan!50!black]
            plot ({\x}, {(1+\x)*(1-\x/2)});
        \node[anchor=east] at (-.0,1) {$1$};
    \end{tikzpicture}
    \begin{tikzpicture}[xscale=2,yscale=2.5,line width=1,>=stealth]
        \def\t{0.2}
        \draw (0,0.6) -- (0,0.8);
        \draw (-.05,0.78) -- (0.05, 0.82);
        \draw (-.05,0.82) -- (0.05, 0.86);
        \draw[->] (0,0.84) -- (0,1.6);
        \begin{scope}[yshift=0.6cm]
            \draw[->] (0,0) -- (1.1,0);
            \draw (\t,-.05) -- (\t,.05);
            \draw (0.5,-.05) -- (0.5,.05);
            \draw (1-\t,-.05) -- (1-\t,.05);
            \draw (1,-.05) -- (1,.05);
            \fill[cyan!50!black,opacity=0.5] (0,-.05) rectangle (\t,.05);
            \fill[green!50!black,opacity=0.5] (\t,-.05) rectangle (0.5,.05);
            \fill[blue!50!black,opacity=0.5] (0.5,-.05) rectangle (1-\t,.05);
            \fill[cyan!50!black,opacity=0.5] (1-\t,-.05) rectangle (1,.05);
        \end{scope}
        \draw[dotted] (0,1) -- (1,1);
        \draw[domain=\t:0.5,smooth, variable=\x, green!50!black]
            plot ({\x}, {(1+\x)*(1+\t/2-\x)});
        \draw[domain=\t:0.5,smooth, variable=\x, blue!50!black]
            plot ({\x}, {(1+1-\x)*(1/2+\t/2)});
        \draw[domain=0.5:1-\t,smooth, variable=\x, green!50!black]
            plot ({\x}, {(1+\x)*(1/2+\t/2)});
        \draw[domain=0.5:1-\t,smooth, variable=\x, blue!50!black]
            plot ({\x}, {(1+1-\x)*(\t/2+\x)});
        \draw[domain=0:\t,smooth, variable=\x, cyan!50!black]
            plot ({\x}, {(1+\x)*(1-\x/2)});
        \draw[domain=1-\t:1,smooth, variable=\x, cyan!50!black]
            plot ({\x}, {(1+\x)*(1-\x/2)});
        \node[anchor=east] at (-.0,1) {$1$};
    \end{tikzpicture}
    \begin{tikzpicture}[xscale=2,yscale=2.5,line width=1,>=stealth]
        \draw (0,0.6) -- (0,0.8);
        \draw (-.05,0.78) -- (0.05, 0.82);
        \draw (-.05,0.82) -- (0.05, 0.86);
        \draw[->] (0,0.84) -- (0,1.6);
        \begin{scope}[yshift=0.6cm]
            \draw[->] (0,0) -- (1.1,0);
            \draw (0.5,-.05) -- (0.5,.05);
            \draw (1,-.05) -- (1,.05);
            \fill[green!50!black,opacity=0.5] (0,-.05) rectangle (0.5,.05);
            \fill[blue!50!black,opacity=0.5] (0.5,-.05) rectangle (1,.05);
        \end{scope}
        \draw[dotted] (0,1) -- (1,1);
        \draw[domain=0:0.5,smooth, variable=\x, green!50!black]
            plot ({\x}, {(1+\x)*(1-\x)});
        \draw[domain=0:0.5,smooth, variable=\x, blue!50!black]
            plot ({\x}, {(1+1-\x)/2});
        \draw[domain=0.5:1,smooth, variable=\x, green!50!black]
            plot ({\x}, {(1+\x)/2});
        \draw[domain=0.5:1,smooth, variable=\x, blue!50!black]
            plot ({\x}, {(1+1-\x)*(\x)});
        \node[anchor=east] at (-.0,1) {$1$};
    \end{tikzpicture}
    \end{center}
    \textbf{(a)} Five different equilibria when the distribution $\mu$ is uniform over the diagonal $\{(x,y)\in T\,|\,x+y=1\}$. Students' utility induce a total ordering over the family of equilibria, where the leftmost equilibrium is student optimal and the rightmost equilibrium is student pessimal.
    \begin{center}
    \begin{tikzpicture}[>=stealth,scale=1.4,line width=1]
        \fill[cyan!50!black,opacity=0.5] (0,0) rectangle (1,1);
        \draw[->] (0,0) -- (1.1,0);
        \draw[->] (0,0) -- (0,1.1);
        \draw (1,0) -- (1,1) -- (0,1);
    \end{tikzpicture}
    \begin{tikzpicture}[>=stealth,scale=1.4,line width=1]
        \begin{scope}[green!50!black,opacity=0.5]
            \fill (0,0) -- (1,1) -- (1,0.5) -- (0.5,0);
            \fill (0,1) -- (0.5,1) -- (0,0.5);
        \end{scope}
        \begin{scope}[blue!50!black,opacity=0.5]
            \fill (0,0) -- (1,1) -- (0.5,1) -- (0,0.5);
            \fill (1,0) -- (1,0.5) -- (0.5,0);
        \end{scope}
        \draw[->] (0,0) -- (1.1,0);
        \draw[->] (0,0) -- (0,1.1);
        \draw (1,0) -- (1,1) -- (0,1);
    \end{tikzpicture}
    \begin{tikzpicture}[>=stealth,scale=1.4,line width=1]
        \begin{scope}[blue!50!black,opacity=0.5]
            \fill (0,0) -- (1,1) -- (1,0.5) -- (0.5,0);
            \fill (0,1) -- (0.5,1) -- (0,0.5);
        \end{scope}
        \begin{scope}[green!50!black,opacity=0.5]
            \fill (0,0) -- (1,1) -- (0.5,1) -- (0,0.5);
            \fill (1,0) -- (1,0.5) -- (0.5,0);
        \end{scope}
        \draw[->] (0,0) -- (1.1,0);
        \draw[->] (0,0) -- (0,1.1);
        \draw (1,0) -- (1,1) -- (0,1);
    \end{tikzpicture}
    \begin{tikzpicture}[>=stealth,scale=1.4,line width=1]
        \begin{scope}[green!50!black,opacity=0.5]
            \fill (0.5,0) rectangle (1,0.5);
            \fill (0,0.5) rectangle (0.5,1);
        \end{scope}
        \begin{scope}[blue!50!black,opacity=0.5]
            \fill (0,0) rectangle (0.5,0.5);
            \fill (0.5,0.5) rectangle (1,1);
        \end{scope}
        \draw[->] (0,0) -- (1.1,0);
        \draw[->] (0,0) -- (0,1.1);
        \draw (1,0) -- (1,1) -- (0,1);
    \end{tikzpicture}
    \begin{tikzpicture}[>=stealth,scale=1.4,line width=1]
        \begin{scope}[blue!50!black,opacity=0.5]
            \fill (0.5,0) rectangle (1,0.5);
            \fill (0,0.5) rectangle (0.5,1);
        \end{scope}
        \begin{scope}[green!50!black,opacity=0.5]
            \fill (0,0) rectangle (0.5,0.5);
            \fill (0.5,0.5) rectangle (1,1);
        \end{scope}
        \draw[->] (0,0) -- (1.1,0);
        \draw[->] (0,0) -- (0,1.1);
        \draw (1,0) -- (1,1) -- (0,1);
    \end{tikzpicture}
    \begin{tikzpicture}[>=stealth,scale=1.4,line width=1]
        \begin{scope}[green!50!black,opacity=0.5]
            \fill (0.5,0.5)--(0,0)--(0,0.5);
            \fill (0.5,0.5)--(0,1)--(0.5,1);
            \fill (0.5,0.5)--(1,1)--(1,0.5);
            \fill (0.5,0.5)--(1,0)--(0.5,0);
        \end{scope}
        \begin{scope}[blue!50!black,opacity=0.5]
            \fill (0.5,0.5)--(0,0)--(0.5,0);
            \fill (0.5,0.5)--(1,0)--(1,0.5);
            \fill (0.5,0.5)--(1,1)--(0.5,1);
            \fill (0.5,0.5)--(0,1)--(0,0.5);
        \end{scope}
        \draw[->] (0,0) -- (1.1,0);
        \draw[->] (0,0) -- (0,1.1);
        \draw (1,0) -- (1,1) -- (0,1);
    \end{tikzpicture}
    \begin{tikzpicture}[>=stealth,scale=1.4,line width=1]
        \begin{scope}[blue!50!black,opacity=0.5]
            \fill (0.5,0.5)--(0,0)--(0,0.5);
            \fill (0.5,0.5)--(0,1)--(0.5,1);
            \fill (0.5,0.5)--(1,1)--(1,0.5);
            \fill (0.5,0.5)--(1,0)--(0.5,0);
        \end{scope}
        \begin{scope}[green!50!black,opacity=0.5]
            \fill (0.5,0.5)--(0,0)--(0.5,0);
            \fill (0.5,0.5)--(1,0)--(1,0.5);
            \fill (0.5,0.5)--(1,1)--(0.5,1);
            \fill (0.5,0.5)--(0,1)--(0,0.5);
        \end{scope}
        \draw[->] (0,0) -- (1.1,0);
        \draw[->] (0,0) -- (0,1.1);
        \draw (1,0) -- (1,1) -- (0,1);
    \end{tikzpicture}
    \begin{tikzpicture}[>=stealth,scale=1.4,line width=1]
        \begin{scope}[scale=.5]
            \fill[cyan!50!black,opacity=0.5] (0,0) rectangle (1,1);
        \end{scope}
        \begin{scope}[green!50!black,opacity=0.5]
            \begin{scope}[scale=.5,shift={(1,1)}]
                \fill (0,0) -- (1,1) -- (0.5,1) -- (0,0.5);
                \fill (1,0) -- (1,0.5) -- (0.5,0);
            \end{scope}
            \begin{scope}[scale=.5,shift={(0,1)}]
                \fill (0.5,0.5) rectangle (1,1);
                \fill (0,0) rectangle (0.5,0.5);
            \end{scope}
            \begin{scope}[scale=.5,shift={(1,0)}]
                \fill (0.5,0.5)--(0,0)--(0,0.5);
                \fill (0.5,0.5)--(0,1)--(0.5,1);
                \fill (0.5,0.5)--(1,1)--(1,0.5);
                \fill (0.5,0.5)--(1,0)--(0.5,0);
            \end{scope}
        \end{scope}
        \begin{scope}[blue!50!black,opacity=0.5]
            \begin{scope}[scale=.5,shift={(1,1)}]
                \fill (0,0) -- (1,1) -- (1,0.5) -- (0.5,0);
                \fill (0,1) -- (0.5,1) -- (0,0.5);
            \end{scope}
            \begin{scope}[scale=.5,shift={(0,1)}]
                \fill (0.5,0) rectangle (1,0.5);
                \fill (0,0.5) rectangle (0.5,1);
            \end{scope}
            \begin{scope}[scale=.5,shift={(1,0)}]
                \fill (0.5,0.5)--(0,0)--(0.5,0);
                \fill (0.5,0.5)--(1,0)--(1,0.5);
                \fill (0.5,0.5)--(1,1)--(0.5,1);
                \fill (0.5,0.5)--(0,1)--(0,0.5);
            \end{scope}
        \end{scope}
        \draw[->] (0,0) -- (1.1,0);
        \draw[->] (0,0) -- (0,1.1);
        \draw (1,0) -- (1,1) -- (0,1);
    \end{tikzpicture}
    \end{center}
    \textbf{(b)} Eight different equilibria when the distribution $\mu$ is uniform over $T$. Notice that for every fixed $x$ (or $y$), the fractions of types applying to schools 1 and 2 is the same. Thus, in each equilibrium, a student of type $(x,y)$ is accepted to school $1$ (resp. $2$) with probability $(1+y)/2$ (resp. $(1+x)/2$) and her expected utility is $(1+x)(1+y)/2$ for both schools.
    \caption{Multiple Bayes-Nash equilibria when $\ell=1$ and $r=1$ in the game described in \cref{fig:example-1.1}. Such phenomenon can be explained by the multiplicity of stable matchings (see panel a), or by a ``purification theorem'' type of argument (see panel b).
    }
    \label{fig:example-1.3}
\end{myfigure}

%%%%%%%%%%%%%%%%%%%%%%%%%%%%%%%%%%%%%%%%%%%%%%%%%%%%%%%%%%%%%%%%%%%%%%%%%%
%%%%%%%%%%%%%%%%%%%%%%%%%%%%%%%%%%%%%%%%%%%%%%%%%%%%%%%%%%%%%%%%%%%%%%%%%%

\begin{myfigure}[p!]
    \centering
    \begin{tikzpicture}[scale=3,line width=1,>=stealth]
        %%%%%
        \node at (0.55,1.1) {\small $s(x,y) = \sqrt{x^2 + y^2}$};
        %\node at (-.8,0.4) {Ranking: $b \succ c \succ a$};
        %%%%%
        \stickfigure{2}{0.9}\node at(2,0.9) {a};
        \stickfigure{2}{0.5}\node at(2,0.5) {b};
        \stickfigure{2}{0.1}\node at(2,0.1) {c};
        \draw (2.2,0.9) -- (2.8,0.3);
        \draw (2.2,0.5) -- (2.8,0.8);
        \draw (2.2,0.1) -- (2.8,0.2);
        %%%%%
        \begin{scope}[color=green!50!black]
            \stickhouse{3}{0.8}\node at(3,0.8) {1};
            \node[anchor=north] at (0.5,0) {math affinity};
            \fill[opacity=0.2] (0,0) -- (1,1) -- (1,0);
            \node[anchor=west] at(3.2,0.95) {$b \succ c \succ a$};
            \node[anchor=west] at(3.2,0.80) {$v_1(x,y) = r+x$};
            \node[anchor=west] at(3.2,0.65) {$c_1 = 1$};
        \end{scope}
        \begin{scope}[color=blue!50!black]
            \stickhouse{3}{0.2}\node at(3,0.2) {2};
            \node[anchor=north, rotate=-90] at (0,0.5) {art affinity};
            \fill[opacity=0.2] (0,0) -- (1,1) -- (0,1);
            \node[anchor=west] at(3.2,0.35) {$b \succ c \succ a$};
            \node[anchor=west] at(3.2,0.20) {$v_2(x,y) = r+y$};
            \node[anchor=west] at(3.2,0.05) {$c_2 = 2$};
        \end{scope}
        \draw[->] (0,0) -- (1.1,0);
        \draw[->] (0,0) -- (0,1.1);
        \draw (1,0) -- (1,1) -- (0,1);
        \foreach \r in {0,0.1,...,1}
            \draw[gray, dashed] (\r,0) arc (0:90:\r);
        \foreach \r in {1,1.1,...,1.5}
            \draw[gray, dashed] (1,{sqrt(\r*\r-1)})
            arc ({acos(1/\r)}:{asin(1/\r)}:\r);
        \fill (20:0.7) circle (0.03);
        \fill (40:1.1) circle (0.03);
        \fill (70:0.8) circle (0.03);
        \node at (15:0.75) {a};
        \node at (36:1.16) {b};
        \node at (75:0.85) {c};
        \node at (1.2,0) {$x$};
        \node at (0,1.2) {$y$};
    \end{tikzpicture}
    \medbreak
    \begin{tikzpicture}[scale=2,line width=1,>=stealth]
        \useasboundingbox (-.2,-.1) rectangle (1.2,1.2);
        \fill[green!50!black,opacity=0.5,
                yshift=1cm,xscale=0.0333,yscale=-0.0333] svg "
                M 1.50,859.00
                C 136.50,780.00 475.50,592.00 585.67,447.33
                653.33,368.00 664.00,208.00 850.00,1.50
                L 850,0 L 850,850 L 0,850";
        \fill[blue!50!black,opacity=0.5,
                yshift=1cm,xscale=0.0333,yscale=-0.0333] svg "
                M 1.50,859.00
                C 136.50,780.00 475.50,592.00 585.67,447.33
                653.33,368.00 664.00,208.00 850.00,1.50
                L 0,0 L 0,850";
        \draw[->] (0,0) -- (1.1,0);
        \draw[->] (0,0) -- (0,1.1);
        \draw (1,0) -- (1,1) -- (0,1);
        \node at (.5,1.1) {$r = 0$};
    \end{tikzpicture}
    \begin{tikzpicture}[scale=2,line width=1,>=stealth]
        \useasboundingbox (-.2,-.1) rectangle (1.2,1.2);
        \fill[green!50!black,opacity=0.5,
                yshift=1cm,xscale=0.0333,yscale=-0.0333] svg "
                M 40.00,859.00
                C 263.00,679.00 519.00,722.00 654.00,537.00
                730.00,436.00 683.00,184.00 848.00,4.00
                L 850,0 L 850,850 L 0,850";
        \fill[blue!50!black,opacity=0.5,
                yshift=1cm,xscale=0.0333,yscale=-0.0333] svg "
                M 40.00,859.00
                C 263.00,679.00 519.00,722.00 654.00,537.00
                730.00,436.00 683.00,184.00 848.00,4.00
                L 850,0 L 0,0 L 0,850";
        \draw[->] (0,0) -- (1.1,0);
        \draw[->] (0,0) -- (0,1.1);
        \draw (1,0) -- (1,1) -- (0,1);
        \node at (.5,1.1) {$r=1$};
    \end{tikzpicture}
    \begin{tikzpicture}[scale=2,line width=1,>=stealth]
        \useasboundingbox (-.2,-.1) rectangle (1.2,1.2);
        \fill[green!50!black,opacity=0.5,
                yshift=1cm,xscale=0.0333,yscale=-0.0333] svg "
                M 0.50 859.50
                C 48.00 825.50 202.00 688.50 369.50 668.00
                537.00 647.50 618.00 661.00 691.00 580.00
                764.00 503.00 709.25 145.25 851.75 1.25
                L 850,0 L 850,850 L 0,850";
        \fill[blue!50!black,opacity=0.5,
                yshift=1cm,xscale=0.0333,yscale=-0.0333] svg "
                M 0.50 859.50
                C 48.00 825.50 202.00 688.50 369.50 668.00
                537.00 647.50 618.00 661.00 691.00 580.00
                764.00 503.00 709.25 145.25 851.75 1.25
                L 850,0 L 0,0 L 0,850";
        \draw[->] (0,0) -- (1.1,0);
        \draw[->] (0,0) -- (0,1.1);
        \draw (1,0) -- (1,1) -- (0,1);
        \node at (.5,1.1) {$r = 2$};
    \end{tikzpicture}
    \begin{tikzpicture}[scale=2,line width=1,>=stealth]
        \useasboundingbox (-.2,-.1) rectangle (1.2,1.2);
        \fill[green!50!black,opacity=0.5,
                yshift=1cm,xscale=0.0333,yscale=-0.0333] svg "
                M 2.00 856.00
                C 41.00 831.00 176.00 655.00 345.00 617.00
                580.00 570.00 668.25 667.00 743.50 609.75
                812.00 549.00 744.25 100.00 849.50 2.25
                L 850,0 L 850,850 L 0,850";
        \fill[blue!50!black,opacity=0.5,
                yshift=1cm,xscale=0.0333,yscale=-0.0333] svg "
                M 2.00 856.00
                C 41.00 831.00 176.00 655.00 345.00 617.00
                580.00 570.00 668.25 667.00 743.50 609.75
                812.00 549.00 744.25 100.00 849.50 2.25
                L 850,0 L 0,0 L 0,850";
        \draw[->] (0,0) -- (1.1,0);
        \draw[->] (0,0) -- (0,1.1);
        \draw (1,0) -- (1,1) -- (0,1);
        \node at (.5,1.1) {$r=5$};
    \end{tikzpicture}
    \begin{tikzpicture}[line width=1,scale=2]
        \useasboundingbox (0,-.2) rectangle (0.8,0.7);
        \fill[green!50!black,opacity=0.5] (0,0.3) rectangle (0.2,0.4);
        \fill[blue!50!black,opacity=0.5] (0,0.1) rectangle (0.2,0.2);
        \draw (0,0.3) rectangle (0.2,0.4);
        \draw (0,0.1) rectangle (0.2,0.2);
        \node[anchor=west] at (-.1,0.55) {Equilibrium:};
        \node[anchor=west] at (0.2,0.35) {school 1};
        \node[anchor=west] at (0.2,0.15) {school 2};
    \end{tikzpicture}
    \caption{Example where preferences of schools are aligned (scoring functions are equal to $s$). The two dimensions of a student's type can be thought as her affinity with maths and with art. Schools sort student by decreasing order of euclidean norm. School 1 is a math school, and is preferred by students with a good affinity with math. School 2 is an art school, and is preferred by students with a good affinity with art. Risk aversion is modeled with the parameter $r \geq 0$. %Figure~\ref{fig:simulations-1} illustrates how the Bayes-Nash equilibria are computed, using the algorithm presented in Section~\ref{sec:computealpha}.
    Each student can apply to $\ell=1$ school, and the strategy of a student having type $(x,y)$ is represented by the color of the point at those coordinates.}
    \label{fig:example-2}
\end{myfigure}
%%%%%%%%%%%%%%%%%%%%%%%%%%%%%%%%%%%%%%%%%%%%%%%%%%%%%%%%%%%%%%%%%%%%%%%%%%
%%%%%%%%%%%%%%%%%%%%%%%%%%%%%%%%%%%%%%%%%%%%%%%%%%%%%%%%%%%%%%%%%%%%%%%%%%

\begin{myfigure}[p!]
    \centering
    \begin{tikzpicture}[scale=3,line width=1,>=stealth]
        %%%%%
        \stickfigure{2}{0.9}\node at(2,0.9) {a};
        \stickfigure{2}{0.5}\node at(2,0.5) {b};
        \stickfigure{2}{0.1}\node at(2,0.1) {c};
        \draw (2.2,0.9) -- (2.8,0.3);
        \draw (2.2,0.5) -- (2.8,0.8);
        \draw (2.2,0.1) -- (2.8,0.2);
        %%%%%
        \draw[->] (0,0) -- (1.1,0);
        \draw[->] (0,0) -- (0,1.1);
        \draw (1,0) -- (1,1) -- (0,1);
        \node[anchor=north] at (0.5,0) {math grade};
        \node[anchor=north, rotate=-90] at (0,0.5) {art grade};
        \begin{scope}[color=green!50!black]
            \draw[dashed] (0.8,0.7) -- (0.65,1);
            \draw[dashed] (0.6,0.4) -- (0.3,1);
            \draw[dashed] (0.2,0.8) -- (0.1,1);
            \node[anchor=south] at (0.5,1) {$s_1(t) = z_1 \cdot t$};
            \node at (0.5, 0.2) {$z_1$};
            \draw[->] (0,0) -- (0.4, 0.2); 
            \stickhouse{3}{0.8}\node at(3,0.8) {1};
            \node[anchor=west] at(3.2,0.95) {$b \succ c \succ a$};
            \node[anchor=west] at(3.2,0.80) {$v_1(x,y) = r+1$};
            \node[anchor=west] at(3.2,0.65) {$c_1 = 1$};
        \end{scope}
        \begin{scope}[color=blue!50!black]
            \draw[dashed] (0.8,0.7) -- (1,0.6);
            \draw[dashed] (0.6,0.4) -- (1,0.2);
            \draw[dashed] (0.2,0.8) -- (1,0.4);
            \node[anchor=south,rotate=-90] at (1,0.5) {$s_2(t) = z_2 \cdot t$};
            \node at (0.2, 0.5) {$z_2$};
            \draw[->] (0,0) -- (0.2, 0.4); 
            \stickhouse{3}{0.2}\node at(3,0.2) {2};
            \node[anchor=west] at(3.2,0.35) {$b \succ a \succ c$};
            \node[anchor=west] at(3.2,0.20) {$v_2(x,y) = r$};
            \node[anchor=west] at(3.2,0.05) {$c_2 = 2$};
        \end{scope}
        \fill (0.8,0.7) circle (0.03);
        \fill (0.6,0.4) circle (0.03);
        \fill (0.2,0.8) circle (0.03);
        \node at (0.1,0.8) {a};
        \node at (0.7,0.7) {b};
        \node at (0.5,0.4) {c};
        \node at (1.2,0) {$x$};
        \node at (0,1.2) {$y$};
    \end{tikzpicture}
    \medbreak
    \begin{tikzpicture}[scale=2,line width=1,>=stealth]
        \useasboundingbox (-.2,-.1) rectangle (1.2,1.2);
        \fill[green!50!black,opacity=0.5,
                yshift=1cm,xscale=0.0403,yscale=-0.0403] svg "
                M 0.09,701.88
                C 212.00,688.00 326.18,698.09 340.00,675.00
                350.91,657.73 354.75,619.92 35.75,0.25
                L 702,0 L 707,702 L 0,702";
        \fill[blue!50!black,opacity=0.5,
                yshift=1cm,xscale=0.0403,yscale=-0.0403] svg "
                M 0.09,701.88
                C 212.00,688.00 326.18,698.09 340.00,675.00
                350.91,657.73 354.75,619.92 35.75,0.25
                L 702,0 L 0,0 L 0,702";
        \draw[->] (0,0) -- (1.1,0);
        \draw[->] (0,0) -- (0,1.1);
        \draw (1,0) -- (1,1) -- (0,1);
        \node at (.5,1.1) {$r = 1/10$};
    \end{tikzpicture}
    \begin{tikzpicture}[scale=2,line width=1,>=stealth]
        \useasboundingbox (-.2,-.1) rectangle (1.2,1.2);
        \fill[green!50!black,opacity=0.5,yshift=1cm,xscale=0.0403,
                yscale=-0.0403] svg "
                M 0.00,702.00
                C 201.33,613.33 465.33,632.00 509.50,533.50
                554.67,432.00 399.33,151.33 322.00,0.33
                L 702,0 L 707,702 L 0,702";
        \fill[blue!50!black,opacity=0.5,yshift=1cm,xscale=0.0403,
                yscale=-0.0403] svg "
                M 0.00,702.00
                C 201.33,613.33 465.33,632.00 509.50,533.50
                554.67,432.00 399.33,151.33 322.00,0.33
                L 702,0 L 0,0 L 0,702";
        \draw[->] (0,0) -- (1.1,0);
        \draw[->] (0,0) -- (0,1.1);
        \draw (1,0) -- (1,1) -- (0,1);
        \node at (.5,1.1) {$r=1$};
    \end{tikzpicture}
    \begin{tikzpicture}[scale=2,line width=1,>=stealth]
        \useasboundingbox (-.2,-.1) rectangle (1.2,1.2);
        \fill[green!50!black,opacity=0.5,
                yshift=1cm,xscale=0.0403,yscale=-0.0403] svg "
                M 0.09,701.73
                C 192.00,535.33 389.70,560.30 520.36,485.64
                553.45,464.00 560.50,456.25 600.00,379.00
                636.00,316.67 656.27,292.33 657.27,252.00
                655.82,206.00 646.00,185.50 551.82,0.36
                L 702,0 L 707,702 L 0,702";
        \fill[blue!50!black,opacity=0.5,
                yshift=1cm,xscale=0.0403,yscale=-0.0403] svg "
                M 0.09,701.73
                C 192.00,535.33 389.70,560.30 520.36,485.64
                553.45,464.00 560.50,456.25 600.00,379.00
                636.00,316.67 656.27,292.33 657.27,252.00
                655.82,206.00 646.00,185.50 551.82,0.36
                L 702,0 L 0,0 L 0,702";
        \draw[->] (0,0) -- (1.1,0);
        \draw[->] (0,0) -- (0,1.1);
        \draw (1,0) -- (1,1) -- (0,1);
        \node at (.5,1.1) {$r = 10$};
    \end{tikzpicture}
    \begin{tikzpicture}[scale=2,line width=1,>=stealth]
        \useasboundingbox (-.2,-.1) rectangle (1.2,1.2);
        \fill[green!50!black,opacity=0.5,
                yshift=1cm,xscale=0.0403,yscale=-0.0403] svg "
                M 0.09,701.91
                C 174.67,517.67 417.67,549.00 533.67,450.33
                559.27,429.09 572.98,393.32 602.25,345.50
                628.50,299.50 690.50,244.00 702.50,105.00
                704.25,82.75 664.75,20.50 655.38,0.50
                L 702,0 L 707,702 L 0,702";
        \fill[blue!50!black,opacity=0.5,
                yshift=1cm,xscale=0.0403,yscale=-0.0403] svg "
                M 0.09,701.91
                C 174.67,517.67 417.67,549.00 533.67,450.33
                559.27,429.09 572.98,393.32 602.25,345.50
                628.50,299.50 690.50,244.00 702.50,105.00
                704.25,82.75 664.75,20.50 655.38,0.50
                L 702,0 L 0,0 L 0,702";
        \draw[->] (0,0) -- (1.1,0);
        \draw[->] (0,0) -- (0,1.1);
        \draw (1,0) -- (1,1) -- (0,1);
        \node at (.5,1.1) {$r=100$};
    \end{tikzpicture}
    \begin{tikzpicture}[line width=1,scale=2]
        \useasboundingbox (0,-.2) rectangle (0.8,0.7);
        \fill[green!50!black,opacity=0.5] (0,0.3) rectangle (0.2,0.4);
        \fill[blue!50!black,opacity=0.5] (0,0.1) rectangle (0.2,0.2);
        \draw (0,0.3) rectangle (0.2,0.4);
        \draw (0,0.1) rectangle (0.2,0.2);
        \node[anchor=west] at (-.1,0.55) {Equilibrium:};
        \node[anchor=west] at (0.2,0.35) {school 1};
        \node[anchor=west] at (0.2,0.15) {school 2};
    \end{tikzpicture}
    \caption{Example where preferences of students are aligned (value functions are constant). The two dimensions of a student's type can be thought as her grades in maths and in art. School 1 has one seat and gives more importance to the math grade. School 2 has two seats and gives more importance to the art grade. Risk aversion is modeled with the parameter $r \geq 0$. %Figure~\ref{fig:simulations-1} illustrates how the Bayes-Nash equilibria are computed, using the algorithm presented in Section~\ref{sec:computealpha}.
    Each student can apply to $\ell=1$ school, and the strategy of a student having type $(x,y)$ is represented by the color of the point at those coordinates.}
    \label{fig:example-3}
\end{myfigure}
%%%%%%%%%%%%%%%%%%%%%%%%%%%%%%%%%%%%%%%%%%%%%%%%%%%%%%%%%%%%%%%%%%%%%%%%%%
%%%%%%%%%%%%%%%%%%%%%%%%%%%%%%%%%%%%%%%%%%%%%%%%%%%%%%%%%%%%%%%%%%%%%%%%%%

\begin{myfigure}[p!]
    \begin{center}
    \begin{tikzpicture}[scale=3,line width=1,>=stealth]
        \draw[->] (0,0) -- (2.1,0);
        \draw (0,-.03) -- (0,.03);
        \draw (2,-.03) -- (2,.03);
        \fill (0.6,0) circle (0.02);
        \fill (1.2,0) circle (0.02);
        \fill (1.6,0) circle (0.02);
        \stickfigure{0.6}{0.2}\node at(0.6,0.2) {a};
        \stickfigure{1.2}{0.2}\node at(1.2,0.2) {b};
        \stickfigure{1.6}{0.2}\node at(1.6,0.2) {c};
        \node at (0,0.1) {0};
        \node at (2,0.1) {1};
        \node at (1,0.5) {$n=3$ students, $\mu$ is uniform over $T = [0,1]$};
        \node[anchor=west] at (2.5,0.52) {$m=3$ identical schools:};
        \node[anchor=west] at (2.6,0.3) {- capacities $c_1=c_2=c_3=1$};
        \node[anchor=west] at (2.6,0.15) {- values $v_1=v_2=v_3=1$};
        \node[anchor=west] at (2.6,0.0) {- scores $s_1(t) = s_2(t) = s_3(t) = t$};
    \end{tikzpicture}
    \end{center}
    Denote $q_{i,j}(t)$ the probability that a student of type $t$ will be rejected from schools $i$ and $j$, because the other two students have types $>t$ and have already been assigned to those two seats. Denote $u(t)$ the expected utility of a student of type $t$ at equilibrium.
    \smallbreak
    \begin{minipage}{.45\linewidth}
    \begin{tikzpicture}[xscale=2.8,yscale=2,line width=1,>=stealth]
        \draw[->] (0,0) -- (2,0);
        \draw[->] (0,0) -- (0,1.2);
        \draw[domain=0:2,smooth, variable=\x, red!50!black]
            plot ({\x}, {1-(2-\x)*(2-\x)/12});
        \node at (-.1,0.66) {$\frac{2}{3}$};
        \node[anchor=east]  at (2,0.5) {$q_{1,2}(t) = q_{1,3}(t) = q_{2,3}(t) = \frac{(1-t)^2}{3}$};
        \node[anchor=east,red!50!black]  at (2,0.8) {$u(t) = 1 - \frac{(1-t)^2}{3}$};
    \end{tikzpicture}
    \smallbreak Equilibrium \#1: every student draws a list uniformly at random, among the 6 possible lists of length 2.
    \end{minipage}\hfill
    \begin{minipage}{.45\linewidth} 
    \begin{tikzpicture}[xscale=2.8,yscale=2,line width=1,>=stealth]
        \draw[->] (0,0) -- (2,0);
        \draw[->] (0,0) -- (0,1.2);
        \draw[domain=0:2,smooth, variable=\x, blue!50!black]
            plot ({\x}, {1-(2-\x)*(2-\x)/16});
        \node at (-.1,0.75) {$\frac{3}{4}$};
        \node[anchor=east] at (1.8,0.5) {$q_{1,3}(t) = q_{2,3}(t) = \frac{(1-t)^2}{4}$};
        \node[anchor=east]  at (1.8,0.2) {$q_{1,2}(t) = \frac{(1-t)^2}{2}$};
        \node[anchor=east,blue!50!black]  at (1.8,0.8) {$u(t) = 1 - \frac{(1-t)^2}{4}$};
    \end{tikzpicture}
    \smallbreak Equilibrium \#2: every student chooses uniformly at random between the lists $(1,3)$ and $(2,3)$.
    \end{minipage}
    \caption{Example with $n=3$ students, $m=3$ identical schools, and $\ell=2$ applications per student. If students had perfect information or where allowed to apply to all 3 schools, they will all be assigned and will always receive a payoff of 1. When students have imperfect information and can only apply to 2 schools, every symmetric equilibrium will leave a student unassigned with positive probability. The intuition behind the improved payoff in equilibrium \#2 is that students privately agree that school 3 is a ``safety choice''.}
    \label{fig:example-4}
\end{myfigure}
%%%%%%%%%%%%%%%%%%%%%%%%%%%%%%%%%%%%%%%%%%%%%%%%%%%%%%%%%%%%%%%%%%%%%%%%%%
%%%%%%%%%%%%%%%%%%%%%%%%%%%%%%%%%%%%%%%%%%%%%%%%%%%%%%%%%%%%%%%%%%%%%%%%%%
\begin{myfigure}[p!]
    \begin{tikzpicture}[scale=2, line width=1,>=stealth]
        \node[anchor=west] at (-.5,2.6) {$n=50$ students, $\mu$ is uniform over $T =[0,1]$};
        \node[anchor=west] at (-.5,2.4) {$m=5$ schools, with identical preferences $s(t) = t$};
        \begin{scope}[green!70!black]
            \stickhouse{0}{2.0}
        \end{scope}
        \begin{scope}[lime!70!black]
            \stickhouse{0}{1.5}
        \end{scope}
        \begin{scope}[yellow!70!black]
            \stickhouse{0}{1.0}
        \end{scope}
        \begin{scope}[orange!70!black]
            \stickhouse{0}{0.5}
        \end{scope}
        \begin{scope}[red!70!black]
            \stickhouse{0}{0.0}
        \end{scope}
        \node at (0,2.0) {1};
        \node at (0,1.5) {2};
        \node at (0,1.0) {3};
        \node at (0,0.5) {4};
        \node at (0,0.0) {5};
        \node[anchor=west] at (0.2,2.15) {value $v_1 = 5$, capacity $c_1 = 5$.};
        \node[anchor=west] at (0.2,1.65) {value $v_2 = 4$, capacity $c_2 = 10$.};
        \node[anchor=west] at (0.2,1.15) {value $v_3 = 3$, capacity $c_3 = 5$.};
        \node[anchor=west] at (0.2,0.65) {value $v_4 = 2$, capacity $c_4 = 10$.};
        \node[anchor=west] at (0.2,0.15) {value $v_5 = 1$, capacity $c_5 = 10$.};
        \node[anchor=west] at (0.2,1.95) {$\rightarrow$ students of rank 1 to 5};
        \node[anchor=west] at (0.2,1.45) {$\rightarrow$ students of rank 6 to 15};
        \node[anchor=west] at (0.2,0.95) {$\rightarrow$ students of rank 16 to 20};
        \node[anchor=west] at (0.2,0.45) {$\rightarrow$ students of rank 21 to 30};
        \node[anchor=west] at (0.2,-.05) {$\rightarrow$ students of rank 30 to 40};
    \end{tikzpicture}
    %%%%
    \includegraphics{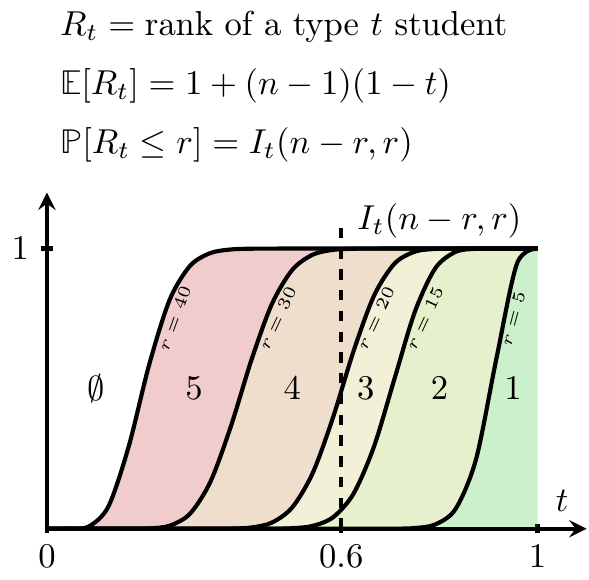}
    \tikzexternaldisable
    %%%%%
    \caption{Example where preferences of schools and students are aligned. If students are allowed to apply to all the schools, their outcome is determined by their rank. Let $R_t$ be the rank of a student of type $t$. Then $R_t-1$ follows a binomial distribution of parameter $(n-1,1-t)$, whose cumulative function can be expressed using the regularized incomplete beta function $I$. In particular, a student of type $t=0.6$ will be assigned to school $3$ with probability $I_{0.6}(30,20)-I_{0.6}(35,15) \approx 43\%$. For every $t$, one could observe that the combined probability of the 2 least likely outcomes never exceeds $2\%$. Hence when students are restricted to 3 applications, each student can, at least $98\%$ of the time, ensure the same outcome as the one when everyone has 5 applications.}
    \label{fig:example-5}
\end{myfigure}
%%%%%%%%%%%%%%%%%%%%%%%%%%%%%%%%%%%%%%%%%%%%%%%%%%%%%%%%%%%%%%%%%%%%%%%%%%
%%%%%%%%%%%%%%%%%%%%%%%%%%%%%%%%%%%%%%%%%%%%%%%%%%%%%%%%%%%%%%%%%%%%%%%%%%

\vspace{-1cm}
\paragraph{Risk aversion.} In \cref{fig:example-1.1}, if students were allowed to apply to both schools, the student proposing deferred acceptance procedure would always choose the student optimal stable matching.
However, for some reason, the clearinghouse only allows students to apply to one school. In order to maximize their expected utilities, students can either prioritize the value they give to schools, or the likelihood of being accepted.
\cref{fig:example-1.2} illustrates two families of Bayes-Nash equilibria, as a function of $r$: the more risk averse students are, the less likely is the student optimal matching to be chosen.

\paragraph{Multiple equilibria.} In the complete information case, every stable matching induces an equilibrium. In the incomplete information case, having multiple stable matchings can induce an infinite number of equilibria. \cref{fig:example-1.3}(a) gives an example with a continuum of equilibria, such that the expected payoff of each student type is non-increasing (from left to right). When multiple stable matchings exist, the left-most equilibrium always implement the student optimal stable matching, and the right-most equilibrium always implement the school optimal stable matching. Conversely, \cref{fig:example-1.3}(b) gives an example with an infinite number of equilibria where each student type receives exactly the same expected utility from every equilibrium.

\subsection{Aligned preferences}

\paragraph{Computing an equilibrium by induction.}
When preferences of schools are identical (scoring functions are equal), the best student can ensure she will be matched with her favorite school. When preferences of students are identical (value functions are constant), the student ranked first by the best school can ensure she will be matched with her favorite school.
An equilibrium for such games can be computed by eliminating dominant strategies: when each type's payoff does not depend on the strategies of types having lower scores, we can proceed by induction. \cref{fig:example-2,fig:example-3} give two examples of such games, and illustrate the type of recommendation one could provide to students using this approach.

\paragraph{Identical schools.} \cref{thm:unique-equilibrium} shows that if the matching market is $\alpha$-reducible and students have perfect information, then every equilibrium yields the same outcome. In the incomplete information case, it is a natural question to ask if the same is true assuming strong $\alpha$-reducibility. \cref{fig:example-4} provides a counter-example with three students and three identical schools. If students can only apply to two schools, one student will stay unassigned with positive probability. One way for the students to reduce this probability is to privately agree that one of the school is a ``safety choice'' which is always ranked last. Even if schools were \textit{a priori} identical, such strategies impact the quality of students selected in the safety school, which may cause a differentiation between the schools from one year to the next. Such an example could be interpreted as a recommendation for identical schools to merge their selection process.

\paragraph{Number of applications.} \cref{fig:example-5} gives an example where both the preferences of students and schools are aligned. In such case there is a unique stable matching where the outcome of a student is determined by her rank: the best students are matched to the best school, the next students are matched to the second school, and so on. A student who has incomplete information and only knows her expected rank can try to apply to all the school in a ``window'' around her expected outcome. Using Hoeffding's inequality, the real rank of a student is at most $x\sqrt{n}$ away from her expected rank, with probability at least $1-2e^{-2x^2}$. If schools have capacity $c$, then an upper quota of $\mathcal O(\sqrt{n}/c)$ applications is already large enough for students to ensure they get the same outcome they would obtain when applying to all the schools, with good probability. Such reasoning could be helpful to decision makers when setting an upper quota on the number of applications.

\FloatBarrier
\section{Existence of a Bayes-Nash equilibrium}
\label{sec:constrained:existence}

\subsection{Induced normal form}
Following Milgrom and Weber \cite{milgrom1985distributional}, we define  distributional strategies as the set $\mathcal{\tilde S}_\mu$ of probability distributions $\tilde p \in \Delta(T\times A)$ such that the marginal distribution on $T$ is the distribution $\mu$. In our setting, behavioral and distributional strategies are equivalent as there is a many-to-one mapping from a behavioral strategy $p$ to the corresponding distributional strategy $\tilde p$.
\begin{itemize}
    \item Given $\mu \in \Delta(T)$ and $p \in \mathcal S$, we define the distribution $\tilde p$ such that  $\tilde p(B \times \{a\}) = \int_B p(t,a) \mathrm d\mu(t)$ for every Borel subset $B \subseteq T$ and for every action $a\in A$.
    \item Conversely, given $\tilde p \in \Delta(T \times A)$, first define the marginal distribution $\mu$ such that $\mu(B) = \tilde p(B \times A)$ for every Borel subset $B \subseteq T$. Then, for every action $a\in A$ the measure $B \mapsto \tilde p(B \times \{a\})$ is absolutely continuous with respect to $\mu$, hence Radon-Nikodym theorem gives the existence of a measurable function $t \mapsto p(t,a)$ such that $\tilde p(B\times \{a\}) = \int_B p(t,a) \mathrm d\mu(t)$ for every $B \subseteq T$.
\end{itemize}
We now define the payoff function $\tilde U$ for distributional strategies. For every $(\tilde p_i)_{i\in[n]} \in \Delta(T\times A)^n$ we set
\begin{equation}\label{eq:payoff}
(\tilde U(\tilde p_i, \tilde p_{-i}))_{i\in[n]} =
\frac{\sum_{a \in A^n}
\int_{t\in T^n}
\textsc{Utility}((t_i,a_i)_{i\in[n]})
\cdot \mathbb 1[\text{all }t_i\text{'s are distinct}] \cdot \prod_{i\in[n]} \mathrm d\tilde p_i(t_i,a_i)}{\sum_{a \in A^n}
\int_{t\in T^n}
\mathbb 1[\text{all }t_i\text{'s are distinct}] \cdot \prod_{i\in[n]} \mathrm d\tilde p_i(t_i,a_i)}
\end{equation}
Notice that the transformation from behavioral to distributional strategies is payoff-preserving, that is
\begin{equation}\label{eq:preserving}
\forall (p_i)_{i\in[n]} \in \mathcal S^n,
\quad (U_\mu(p_i, p_{-i}))_{i\in[n]} = (\tilde U(\tilde p_i, \tilde p_{-i}))_{i\in[n]}.
\end{equation} 
The \emph{induced normal form} of the Bayesian game is the surrogate symmetric $n$ players game $\mathcal G_\mu = \langle n, \mathcal{\tilde S}_\mu, \tilde U\rangle$, where the set of actions is the set of distributional strategies $\mathcal{\tilde S}_\mu$.
Thus, a behavioral strategy profile $(p_i)_{i\in[n]}$ is a mixed Bayes-Nash equilibrium in the original Bayesian game if and only if the corresponding distributional strategy profile $(\tilde p_i)_{i\in[n]}$ is a pure Nash equilibrium in $\mathcal G_\mu$.

The notion of $\varepsilon$-equilibrium in the induced normal form game $\mathcal G_\mu$ exactly corresponds to \emph{ex-ante} $\varepsilon$-equilibrium in the Bayesian game, that is a strategy profile where no student can deviate and win more than $\varepsilon$, in average before drawing her type.

\subsection{Existence theorem}

We are now ready to prove the existence of an equilibrium for the game $\mathcal G_\mu$, and thus for the Bayesian game.

\begin{theorem}\label{thm:existence}
  The game $\mathcal G_\mu$ has a symmetric equilibrium.
\end{theorem}
\begin{proof}
We apply Proposition 1, Proposition 3, and Theorem 1 from \cite{milgrom1985distributional}. Because the set of actions $A$ is finite, the game has equicontinuous payoffs (R1). Types of students are drawn without replacement from~$\mu$, which is equivalent with sampling types independently and condition on the fact that types are distinct. Thus the distribution over $T^n$ (types drawn without replacement) is absolutely continuous with respect to the product distribution (types drawn independently), which proves that the game has absolutely continuous information (R2). Milgrom and Weber use Glicksberg's Theorem to prove that the best response correspondence has a fixpoint, which gives the existence of a Nash equilibrium. Proving the existence of a symmetric equilibrium only requires a small modification of the best response correspondence, see for example \cite{cheng2004notes}.
\end{proof}

\subsection{Computability in the finite case}
When the distribution $\mu$ is discrete and has a finite support $\{t_1, \dots, t_k\}$ of size $k \geq n$, we define the \emph{symmetric agent-form} game $\mathcal G_{\mu}' = \langle k, A, V\rangle$, where each player corresponds to a type, and the payoff function $V: A^k \rightarrow \mathbb R_+^k$ is a vector-valued function, such that the $i$-th coordinate of $V(a)$ is equal to the expected payoff of a student of type $t_i$ in the Bayesian game (the expectation is taken over the types of the $n-1$ other students) when a player of type $t_j$ plays $a_j$. By construction, $p\in \mathcal S$ is a symmetric equilibrium of the Bayesian game if and only if $(p(t_i))_{i\in[k]} \in \Delta(A)^k$ is a mixed equilibrium of the agent-form game.

The notion of $\varepsilon$-equilibrium in the symmetric agent-form game $\mathcal G'_\mu$ exactly corresponds to \emph{interim} $\varepsilon$-equilibrium in the Bayesian game, where no student can deviate and win more that $\varepsilon$ after drawing her type (but before drawing the types of other players). Because an interim approximate equilibrium is also an ex-ante approximate equilibrium, any $\varepsilon$-equilibrium of the game $\mathcal G'_\mu$ induces an $\varepsilon$-equilibrium of the game $\mathcal G_\mu$.

\begin{theorem}
For the game $\mathcal G_\mu'$, computing an exact equilibrium is in the class \textsc{FIXP}, and computing an approximate $\varepsilon$-equilibrium with $\varepsilon > 0$ is in the class \textsc{PPAD}.
\end{theorem}
\begin{proof}
The game $\mathcal G_\mu'$ is a $k$-players game in normal form, where the payoff function $V$ is given by a matrix of size $k\cdot |A|^k$. As such, the problems of computing exact and approximate equilibria are respectively in the classes \textsc{FIXP} and \textsc{PPAD}, see for example~\cite{yannakakis2009equilibria}.
%The existence of a mixed equilibrium is guaranteed by Nash's Existence Theorem~\cite{nash1951non}, which uses Brouwer's Fixed-point Theorem. Computing a fixed-point is in the complexity class \textsc{FIXP}, which is included in \textsc{PSPACE}
\end{proof}

\section{Efficient computation of the equilibrium}
\label{sec:constrained:efficient}

In this section we provide efficient algorithms to compute an equilibrium when $\mu$ has a finite support. The student-proposing deferred acceptance procedure being quite complex, we will make some assumption on the preferences of schools and students in order to simplify the matching procedure. More precisely, \cref{algo:alpha} simplifies the function \textsc{Utility} when the matching market induced by $(t_i,a_i) \in (T\times A)^n$ is $\alpha$-reducible.

\begin{algorithm}[h!]
    \begin{algorithmic}
        \State Game parameters: $n$, $m$, $A$, $T=[0,1]^d$, $(s_j)_{j\in[m]}$, $(v_j)_{j\in[m]}$ and $(c_j)_{j\in[m]}$.
        \State \textbf{Initialization:}
        \State \quad\; Add a ``sentinel'' school $0$ with capacity $c_0 = n$ and value $v_0 = 0$, ranked last in every list.
        \Function{Utility}{$(t_i, a_i)_{i\in[n]}\in (T\times A)^n$}
            \State Each school $j$ sorts its applicants ($i\in [n]$ such that $j \in a_i$) by decreasing score ($s_j(t_i)$).
            \State Initialize $r \gets (c_j)_{0 \leq j \leq m}$ the vector of remaining capacities.
            \While{some students are unassigned}
                \For{each school $j$ with a positive capacity ($r_j > 0$) and some unassigned applicant}
                    \State Let $i$ be the top unassigned applicant at school $j$.
                    \If{school $j$ is student $i$'s first choice among schools with a positive capacity}
                        \State Assign student $i$ to school $j$, set $u_i \gets v_j(t_i)$ and $r_j \gets r_j - 1$.
                    \EndIf
                \EndFor
            \EndWhile
            \State\textbf{Return} the vector of utilities $(u_i)_{i\in[n]}$.
        \EndFunction
    \end{algorithmic}
    \caption{Simplified procedure when the matching market is $\alpha$-reducible.}
    \label{algo:alpha}
\end{algorithm}

\begin{theorem}\label{thm:computealpha}
   If the matching market induced by $(t_i,a_i)_{i\in [n]} \in (T\times A)^n$ is $\alpha$-reducible and schools give distinct scores to student, then \cref{algo:alpha} is equivalent with \cref{algo:game-bayesian}.
\end{theorem}
\begin{proof}
Assuming that schools give distinct scores to students, the  algorithm is able to sort applicants by decreasing scores. Assuming that the matching market is $\alpha$-reducible, at least one student will be assigned at each iteration of the while loop and the algorithm will terminate.
To show the equivalence with \cref{algo:game-bayesian}, consider the first pair $(i,j)$ assigned by \cref{algo:alpha}. By construction, $i$ and $j$ prefer each other to everyone else, thus must be matched in every stable matching, and in particular the one computed by the student proposing deferred acceptance algorithm. We continue with the same reasoning by induction.
\end{proof}

Given a sequence of types $(t_i)_{i \geq 1}$, such that each school gives distinct scores to each type, we denote $\mu_k$ the uniform distribution over $(t_i)_{i\in[k]}$. \cref{sec:constrained:computealpha,sec:constrained:computeidentical} study different hypotheses under which we are able to compute an equilibrium $\tilde p_k$ for the game $\mathcal G_{\mu_k}$.

\subsection{One application per student and strong $\alpha$-reducibility}
\label{sec:constrained:computealpha}

In the complete information setting, and when the matching market is $\alpha$-reducible, \cref{thm:unique-equilibrium} shows that every equilibrium implements the unique stable matching, and that one can build simple equilibrium by induction: there exist a fixed student-school pair $(i,j)$ who prefer each other to everyone else, thus reporting $j$ is a dominant strategy for player $i$, who can be removed from the market (together with her seat), and so on. In the incomplete information case, we build on this intuition, replacing a student-school pair by a type-school pair $(t,j)\in T\times [m]$.

\begin{definition}[strong $\alpha$-reducibility]
\label{def:strongalpha}
We say that the matching game is \emph{strongly $\alpha$-reducible}, where for every $S \subseteq T$ and $q\in[0,1]^m$, there must be at least one pair $(t,j)\in S\times[m]$ such that $s_j(t') \leq s_j(t)$ for every $t'\in S$ and $q_{j'} \cdot v_{j'}(t) \leq q_j\cdot v_j(t)$ for every $j'\in [m]$.
\end{definition}

As a special case, notice that if schools have identical preferences (all functions $s_j$ are equal), or if students have identical preferences (all functions $v_j$ are constant), or if preferences of students and schools are symmetric ($s_j = v_j$ for all $j$), then the game is strongly $\alpha$-reducible.

\begin{algorithm}[h!]
    \begin{algorithmic}
        \State Game parameters: $k$, $n$, $m$, $A = [m]$, $T=[0,1]^d$, $(s_j)_{j\in[m]}$, $(v_j)_{j\in[m]}$ and $(c_j)_{j\in[m]}$.
        \Function{Proba}{$x$, $c$}
            \State\textbf{Return}
            $\sum_{i=0}^{c-1} \binom{x}{i}\binom{k-1-x}{n-1-i}/\binom{k-1}{n-1}$,
            that is, the probability that at most $c-1$
            \\\qquad of the other $n-1$ students will draw one of the $x$ types that 
            \\\qquad have (already) been assigned to this school of capacity $c$.
        \EndFunction
        \Function{Equilibrium}{collection of types $(t_i)_{i\in[k]} \in T^k$}
            \State Initialize $x \gets (0)_{j\in [m]}$ the number of types applying to each school.
            \While{some types are unassigned}
                \For{each school $j \in [m]$}
                \State Let $i$ be an unassigned type which maximizes $s_j(t_i)$.
                    \If{$v_{j'}(t_i) \cdot \textsc{Proba}(x_{j'}, c_{j'}) \leq v_{j}(t_i) \cdot \textsc{Proba}(x_{j}, c_{j})$ for every other $j'\in [m]$}
                    \State Assign type $i$ to school $j$, set $a_i \gets j$ and $x_j \gets x_j + 1$.
                    \EndIf
                \EndFor
            \EndWhile
            \State\textbf{Return} the distributional strategy $\mathrm{Uniform}\{(t_i, a_i)\}_{i \in [k]}$.
        \EndFunction
    \end{algorithmic}
    \caption{Compute an equilibrium with $\ell = 1$ and strong $\alpha$-reducibility.}
    \label{algo:strongalpha}
\end{algorithm}

\begin{theorem}
    \label{thm:compute-strongalpha}
    Let $\mu_k$ be the uniform distribution over $(t_i)_{i \in [k]}$.
    If every school gives distinct scores to types, and if the game is strongly $\alpha$-reducible, then \cref{algo:strongalpha} returns a symmetric equilibrium of the game $\mathcal G_{\mu_k}$.
\end{theorem}

\begin{proof}
    First, we show that if the game is strongly $\alpha$-reducible then \cref{algo:strongalpha} terminates: at each iteration of the while loop there is at least one fixed pair with $q = (\textsc{Proba}(x_j, c_j))_{j \in [m]}$.
    Then we compute a symmetric equilibrium via the elimination of dominant strategies.
    
    At each iteration, we define $q_j = \textsc{Proba}(x_j, c_j)$. Consider type-school pairs $(t,j)$ where school $j$ ranks~$t$ first among types that have not been assigned yet (there are no ties by assumption). Then, a student of type $t$ will be accepted to school $j$ with probability $=q_j$ (because she is ranked first in school $j$), and will be accepted to another school $j'$ with probability $\leq q_{j'}$ (because she might not be ranked first at school $j'$). If $(t,j)$ is a fixed pair, then $q_{j'} \cdot v_{j'} \leq q_j \cdot v_j$ for every $j'$, and it is a dominant strategy for a student of type $t$ to apply to school $j$.
    
    A crucial detail is that student's types are drawn without replacement. This removes any feedback the strategy of a type may have on itself because of multiple students having the same type. This also explains why \textsc{Proba} uses an hypergeometric distribution rather than a simpler binomial distribution.
\end{proof}

\subsection{Schools have identical preferences}
\label{sec:constrained:computeidentical}

Gusflied and Irving \cite{gusfield1989stable} observed that the matching is unique when all schools have identical preferences. In such a case, the matching procedure of \cref{algo:alpha} further simplifies into the serial dictatorship mechanism: the best student chooses her favorite school, then the second best student chooses among remaining schools, and so on.

\begin{algorithm}
    \begin{algorithmic}
        \State \textbf{Game parameters:} $n$, $m$, $A$, $T=[0,1]^d$, $s$, $(c_j)_{j \in [m]}$ and $(v_j)_{j\in [m]}$.
        \State \textbf{Initialization:}
        \State \quad\; Add a ``sentinel'' school $0$ with capacity $c_0 = n$ and value $v_0 = 0$, ranked last in every list.
        \State \quad\; Define the set of ``remaining capacity'' vectors $R = \prod_{0 \leq j \leq m} \{0, 1, \dots, c_j\}$.
        \Function{School}{preference list $a\in A$, remaing capacities $r\in R$}
        \State \textbf{Return} the first school $j$ in the preference list $a$ whose remaining capacity is $r_j > 0$.
        \EndFunction
        \Function{Equilibrium}{collection of types $(t_i)_{i\in[k]} \in T^k$}
        \State Sort types $(t_i)_{i\in[k]}$ by non-increasing order of score $s(t_i)$.
        \State Initialize the distribution $q \in \Delta(R)$ such that $q((c_j)_{0 \leq j \leq m}) = 1$.
        \For{$i$ from $1$ to $k$}
            \State Let $a_i = \argmax_{a\in A} \sum_{r\in R}q(r) \cdot v_{\textsc{school}(a,r)}(t_i)$.
            \For{each ``remaining capacity'' vector $r \in R$ in lexicographical order}
                \State Let $p \gets (n-1-\sum_{0 \leq j \leq m} (c_j - r_j))/(n-i+1)$ be the probability
                \\\hspace{2cm} that a student has type $t_i$.
                \State Let $r' \gets r - \delta_j$ be the capacity vector once a student
                \\\hspace{2cm} is assigned to school $j = \textsc{School}(a_i, r)$.
                \State Set $q(r) \gets (1-p) \cdot q(r)$ and $q(r') \gets q(r') + p \cdot q(r)$.
            \EndFor
        \EndFor
        \State \textbf{Return} the distributional strategy $\mathrm{Uniform}\{(t_i, a_i)\}_{i \in [k]}$.
        \EndFunction
    \end{algorithmic}
    \caption{Computing an equilibrium when schools have identical preferences $s_j = s$.}
    \label{algo:identical}
\end{algorithm}

\begin{theorem}
    \label{thm:compute-identical}
    Let $\mu_k$ be the uniform distribution over $(t_i)_{i \in [k]}$. If all school have the same scoring function~$s$, and if all $s(t_i)$'s are distincts, then \cref{algo:identical} returns a symmetric equilibrium of the game $\mathcal G_{\mu_k}$.
\end{theorem}

\begin{proof}
    We compute a symmetric equilibrium by eliminating dominant strategies. After sorting types by decreasing scores, notice that the expected payoff of a student having type $t_i$ does not depend on the strategy of students having types $t_{i+1}, \dots, t_k$.
    
    At each iteration $i$, $q$ corresponds to the distribution over remaining seats when the serial dictatorship mechanism consider a type $t_i$ student (and assigns her to the first available school in her preference list). Notice that because we allow students to apply to more than 1 school, correlations may exist between the number of remaining seats in different schools, which is the reason why we store the whole distribution and not only its marginals. Then, we compute the expected payoff of each action, chose the best one, and update distribution $q$ accordingly.
    
    As in \cref{thm:compute-strongalpha}, types of students are drawn without replacement. This is exactly the distribution we consider when updating $q$: conditioning on the fact that remaining seats are given by $r$, exactly $x = \sum_{0 \leq j \leq m} (c_j - r_j)$ students have drawn types in $\{t_1, \dots, t_{i-1}\}$, hence $n-1-x$ other students have types in $\{t_i, \dots, t_n\}$, and one of them will draw type $t_i$ with probability $(n-1-x)/(n-i+1)$. 
\end{proof}

\section{Convergence theorem}\label{sec:constrained:approx}

This section is a toolbox to prove that an algorithm approximates a Bayes Nash equilibrium.
\cref{fig:simulations-1,fig:simulations-2} illustrate how one can combine algorithms from \cref{sec:constrained:efficient} with the convergence theorem to compute equilibria of games with a continuous distribution over types.

\paragraph{Weak convergence of type distribution.} Computing an equilibrium is more tractable when the set of strategies has finite dimension. For that matter, when $\mu$ is continuous, we will discretize the set of types.
Denote $\Delta^d(T) \subseteq \Delta(T)$ the set of discrete distributions having a finite support. \cref{thm:approx-distrib} states that $\Delta^d(T)$ is dense in $\Delta(T)$ for the weak convergence of measures. More precisely, one can approximate a distribution by drawing a finite number of independent samples from it.

\begin{theorem}
    \label{thm:approx-distrib}
    Let $\mu\in\Delta(T)$ be a distribution and let $(t_i)_{i\geq 1}$ be a sequence of independent random variables with distribution $\mu$. For all $k \geq 1$, define the (random) distribution $\mu_k$ such that $\mu_k(B) = |\{t_i\}_{i\in[k]}\cap B|/k$ for all Borel set $B \subseteq T$.
     Then almost surely (over the randomness of the $t_i$'s), the sequence $(\mu_k)_{k\geq 1}$ weakly converges towards $\mu$.
\end{theorem}
\begin{proof}
See \cite{varadarajan1958convergence}.
\end{proof}

\paragraph{Weak convergence of distributional strategies.} Let us explain why the formalism of distributional strategies is required. For approximation purposes, we are interested in the case where each $\mu_k$ is discrete and has finite support, converging weakly towards a continuous distribution $\mu$. For any behavioral strategy $p \in \mathcal S$, one could set each $p_k$ to be equal to $p$ almost everywhere (outside of the support of $\mu_k$), while ensuring that each $p_k$ is a symmetric Nash equilibrium under $\mu_k$. Hence, assuming the weak convergence of behavioral strategies is not enough to prove that the limit is a Nash equilibrium.

\paragraph{Continuity of the payoff function.} Unfortunately, assuming the weak convergence of a sequence of equilibria in distributional strategies, is still not enough to show that the limit is an equilibrium. In particular, we cannot directly apply Theorem 2 from Milgrom and Weber \cite{milgrom1985distributional}, because the payoff function $\textsc{Utility}: (T\times A)^n\rightarrow \mathbb R_+^n$ might be discontinuous in the players types. However, we will be able to show the weaker property that $\tilde U: \Delta(T\times A)^n \rightarrow \mathbb R_+$ is (sequentially) continuous at the limit (when distributional strategies are endowed with the topology of the weak convergence of measures). This requires additional continuity assumptions on $\mu$, $s_j$'s and $v_j$'s.

\begin{theorem}
    \label{thm:continuous-payoff}If the following conditions hold, then the utility function $\tilde U$ is weakly continuous at every strategy profile in $(\mathcal{\tilde S}_\mu)^n$.
    \begin{itemize}
        \item the distribution $\mu$ is atomless (that is, $\mu(\{t\}) = 0$ for every $t\in T$),
        \item  value and scoring functions are continuous $\mu$-almost everywhere (that is, $\mu(D) = 0$ where $D$ is the set of discontinuities of a scoring function $s_j$ or of a value function~$v_j$),
        \item level sets of scoring functions are $\mu$-negligible (that is, $\mu(s_j^{-1}(\{y\})) = 0$ for every $j\in[m]$ and $y\in[0,1]$).
    \end{itemize}
\end{theorem}
\begin{proof}
Let $(\tilde p_i)_{i\in[n]} \in (\mathcal{\tilde S}_\mu)^n$ be a strategy profile.
In \cref{eq:payoff} which defines $\tilde U$, we are going to show that integrands are continuous almost everywhere with respect to $(\tilde p_i)_{i\in[n]}$. Then, we conclude the proof using the Portmanteau Theorem (see for example Theorem 3.10.1 from \cite{durrett2019probability}), showing that $\tilde U$ is sequentially\footnote{Using Prokhorov's theorem, the space of probability measures $\Delta(T\times A)$ endowed with its weak topology is metrizable, thus the notions of sequential continuity and continuity are equivalent.} continuous in $(\tilde p_i)_{i\in[n]}$. 

First, assuming that $\mu$ is atomless is sufficient to prove that $\mathbb 1[\text{all }t_i\text{'s are distinct}]$ is continuous almost everywhere.
Moreover, if $\textsc{Utility}$ is not continuous in $(t_i,a_i)_{i\in[n]}$, then it is because a value function or a scoring function is discontinuous in some $t_i$, or because a school gives the same score to two types. Each condition occurs with probability 0, thus $\textsc{Utility}$ is continuous almost everywhere.
\end{proof}

\paragraph{Putting everything together.} To approximate a Nash equilibrium of the game $\mathcal G_\mu$, first use \cref{thm:continuous-payoff} to show that $\tilde U$ is weakly continuous at every strategy profile in $(\mathcal{\tilde S}_\mu)^n$.
Then use \cref{thm:approx-distrib} to build a discrete approximation $\mu_k$ of the type distribution $\mu$. Then, compute a symmetric Nash equilibrium $\tilde p_k$ of the game $\mathcal G_{\mu_k}$. Using Prokhorov's theorem, the set of distributional strategies $\Delta(T\times A)$ is metrizable and (sequentially) compact, hence one can build a converging subsequence of distributional strategies, whose limit will be a symmetric Nash equilibrium of $\mathcal G_\mu$.

\begin{theorem}\label{thm:convergence}
    Consider a sequence of measures $\mu_{k\geq n} \in \Delta(T)$ and a sequence of behavioral strategies $p_{k\geq n} \in \mathcal S$, if 
    \begin{itemize}
        \item for all $k\geq n$, the distributional strategy $\tilde p_k \in \mathcal{\tilde S}_{\mu_k}$ is a symmetric equilibrium of $\mathcal G_{\mu_k}$,
        \item the sequence of distributional strategies weakly converges towards a strategy $\tilde p \in \mathcal{\tilde S}_\mu$ with a marginal type distribution $\mu\in \Delta(T)$,
        \item the payoff function $\tilde U$ is weakly continuous at every strategy profile in $(\mathcal{\tilde S}_\mu)^n$,
    \end{itemize}
    then $\tilde p$ is a symmetric equilibrium for the game $\mathcal G_\mu = \langle n, \mathcal{\tilde S}_\mu, \tilde U\rangle$. Alternatively, if $\tilde p_k$'s are $\varepsilon$-equilibria with $\varepsilon > 0$ (ex-ante approximate equilibria of the Bayesian games), then $\tilde p$ is  an $\varepsilon$-equilibrium of the game $\mathcal G_\mu$.
\end{theorem}

\begin{proof}
For the sake of contradiction, assume that $\tilde p$ is not a symmetric Nash equilibrium of $\mathcal G_\mu$. Then there exists a best response $\tilde p^* \in \mathcal{\tilde S}_\mu$ such that playing $\tilde p^*$ raises the payoff of the player by a positive constant $\varepsilon > 0$, that is $\tilde U(\tilde p^*, (\tilde p)_{i\in[n-1]}) - \tilde U(\tilde p, (\tilde p)_{i\in[n-1]}) = \varepsilon > 0$. The sequence of measures $(\mu_k)_{k\geq 0}$ weakly converges towards $\mu$, hence there exists a sequence of distributional strategies $\tilde p_k^* \in \mathcal{\tilde S}_{\mu_k}$ weakly converging towards $\tilde p^*$. Using the continuity hypothesis on $\tilde U$, we show that $\tilde U(\tilde p_k^*, (\tilde p_k)_{i\in [n-1]}) - \tilde U(\tilde p_k, (\tilde p_k)_{i\in [n-1]})$ converges towards $\varepsilon$, and thus is positive for some $k \geq n$.
Therefore, it contradicts the fact that each $\tilde p_k$ is an equilibrium of $\mathcal G_{\mu_k}$. The proof with approximate equilibria is identical, if we set $\varepsilon$ in the proof to be equal to $\varepsilon$ from the statement of the theorem.
\end{proof}

%%%%%%%%%%%%%%%%%%%%%%%%%%%%%%%%%%%%%%%%%%%%%%%%%%%%%%%%%%%%%%%%%%%%%%%%%%%
%%%%%%%%%%%%%%%%%%%%%%%%%%%%%%%%%%%%%%%%%%%%%%%%%%%%%%%%%%%%%%%%%%%%%%%%%%%

\begin{myfigure}[p!]
    \centering
    \begin{tabular}{ccccc}
        & $r=0$ & $r = 1$ & $r = 2$ & $r = 5$ \\
        \rotatebox{90}{\qquad$k=100$} &
        \includegraphics[width=.2\linewidth]{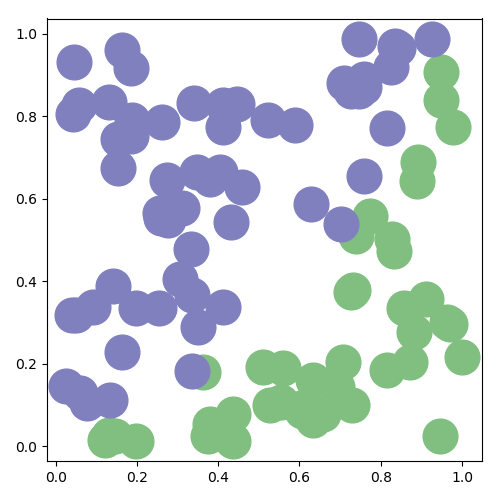} &
        \includegraphics[width=.2\linewidth]{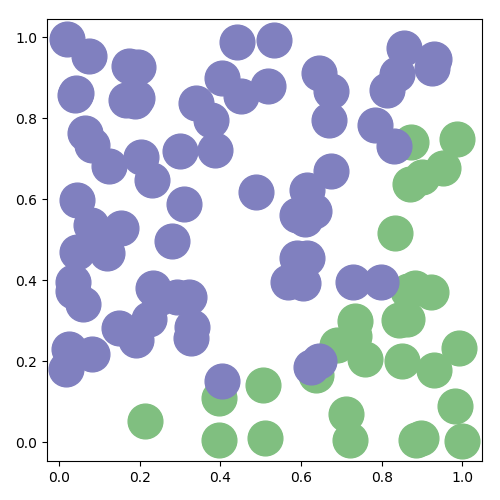} &
        \includegraphics[width=.2\linewidth]{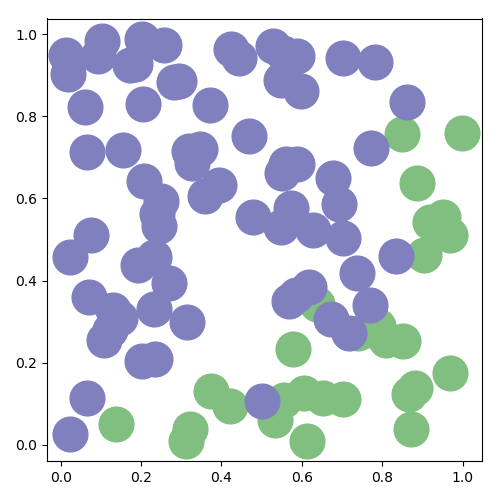} &
        \includegraphics[width=.2\linewidth]{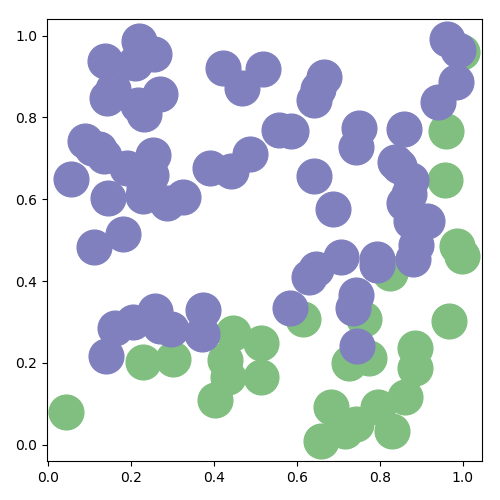} \\
        \rotatebox{90}{\qquad$k=1000$} &
        \includegraphics[width=.2\linewidth]{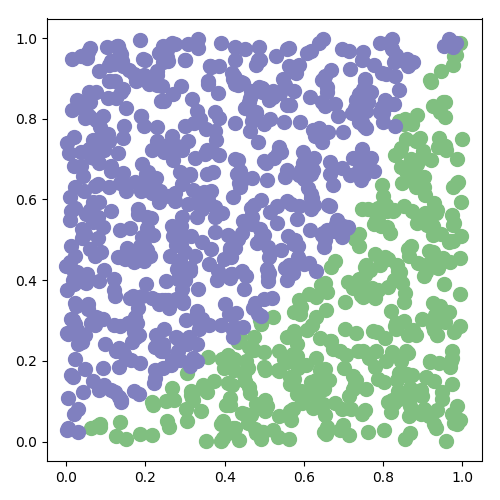} &
        \includegraphics[width=.2\linewidth]{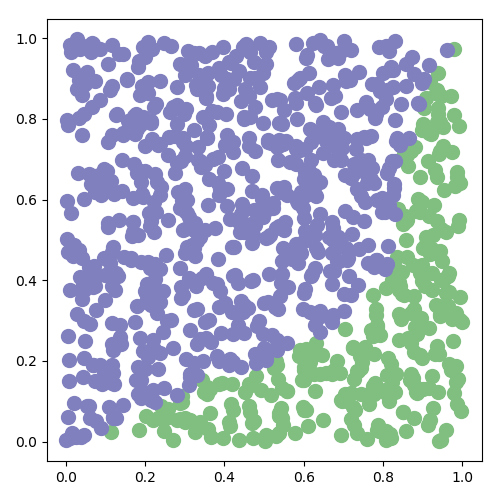} &
        \includegraphics[width=.2\linewidth]{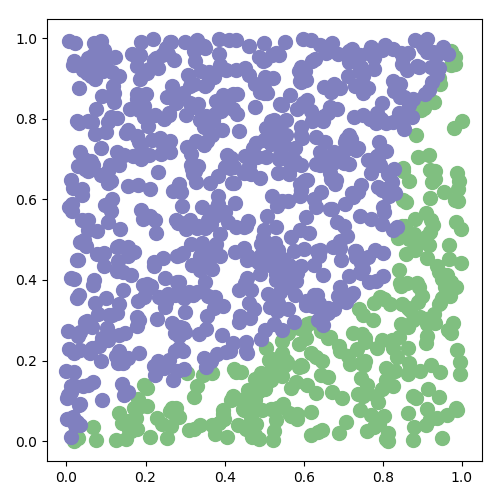} &
        \includegraphics[width=.2\linewidth]{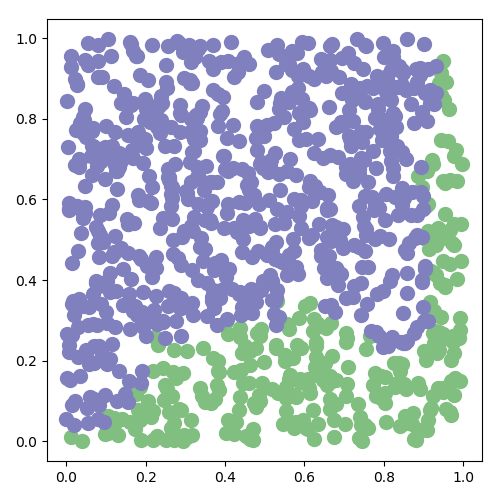} \\
        \rotatebox{90}{\qquad$k=10000$} &
        \includegraphics[width=.2\linewidth]{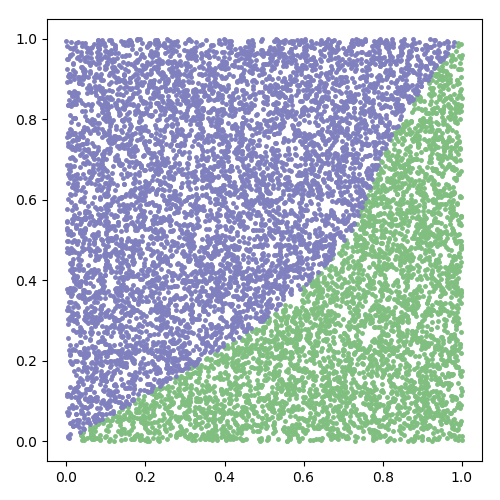} &
        \includegraphics[width=.2\linewidth]{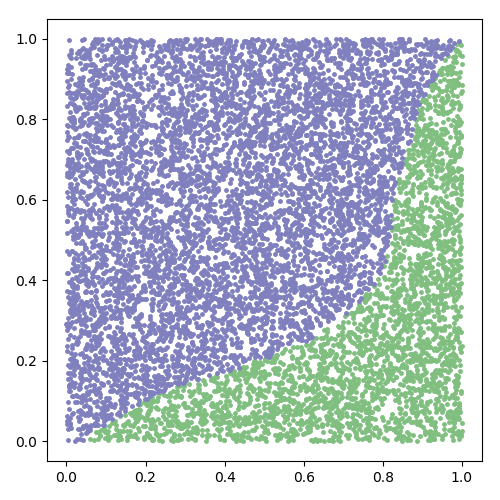} &
        \includegraphics[width=.2\linewidth]{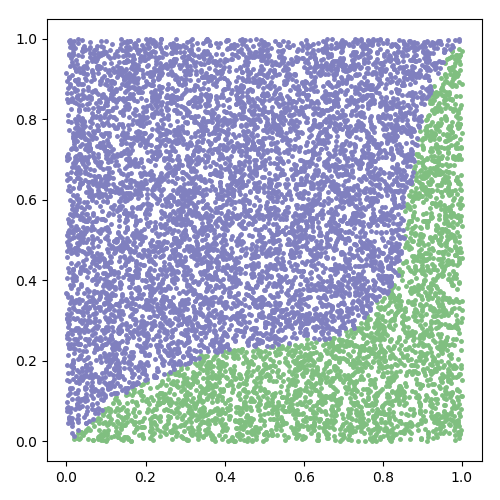} &
        \includegraphics[width=.2\linewidth]{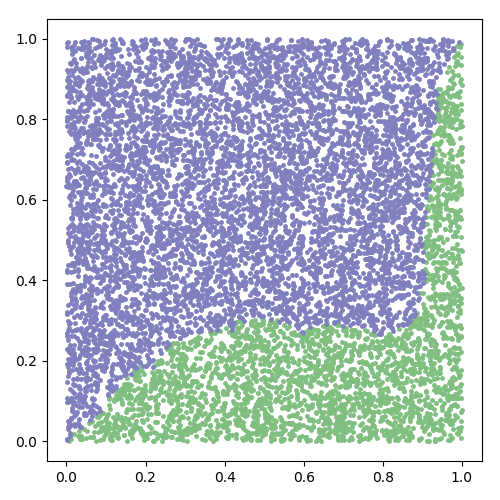} \\
    \end{tabular}
    \par\textbf{(a)} Pure equilibria for random discretizations of the game from \cref{fig:example-2}.
    \bigbreak
    \begin{tabular}{ccccc}
        & $r=0.1$ & $r = 1$ & $r = 10$ & $r = 100$ \\
        \rotatebox{90}{\qquad$k=100$} &
        \includegraphics[width=.2\linewidth]{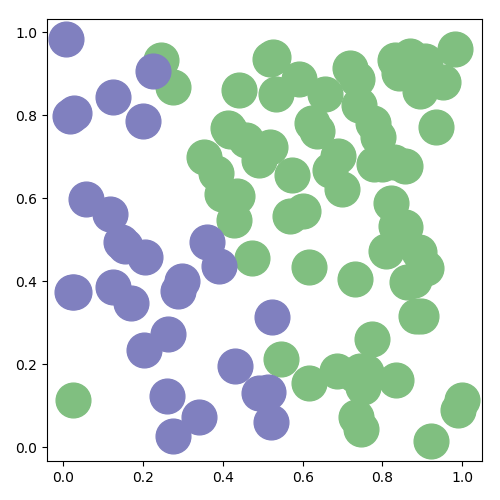} &
        \includegraphics[width=.2\linewidth]{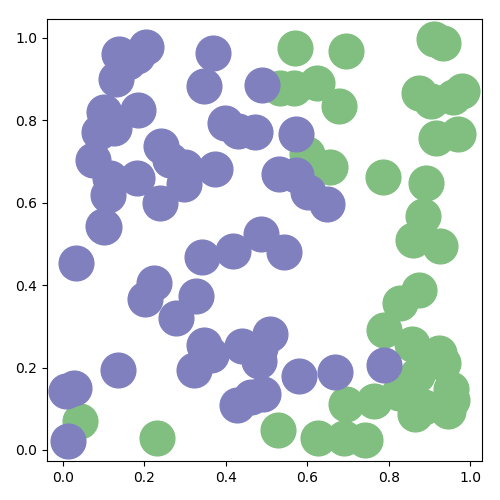} &
        \includegraphics[width=.2\linewidth]{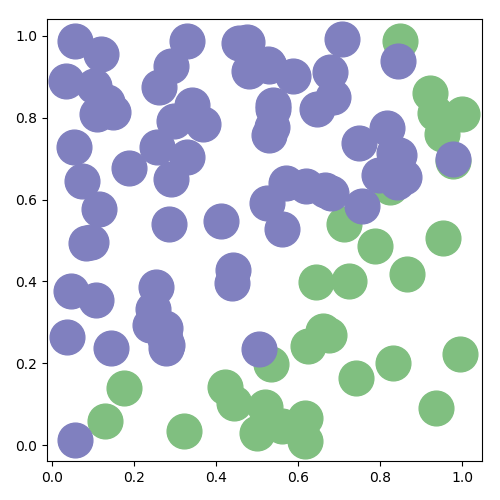} &
        \includegraphics[width=.2\linewidth]{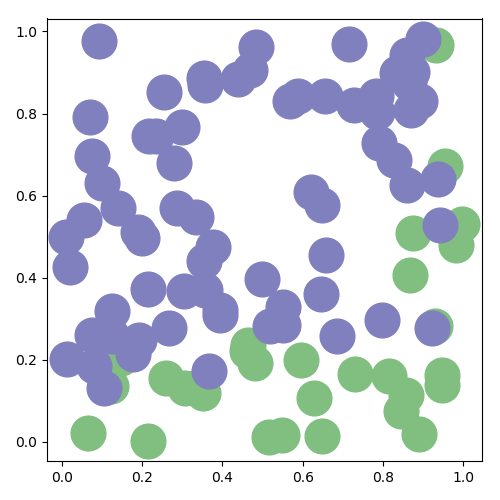} \\
        \rotatebox{90}{\qquad$k=1000$} &
        \includegraphics[width=.2\linewidth]{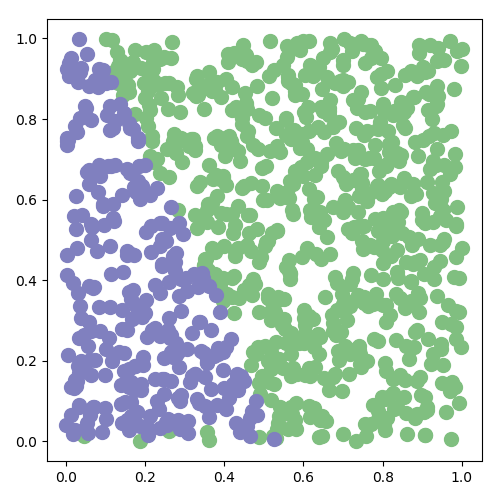} &
        \includegraphics[width=.2\linewidth]{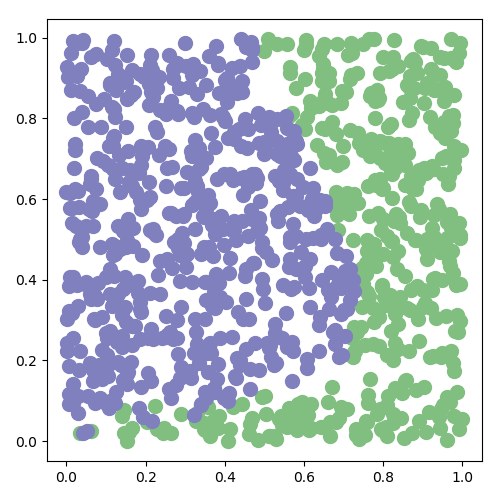} &
        \includegraphics[width=.2\linewidth]{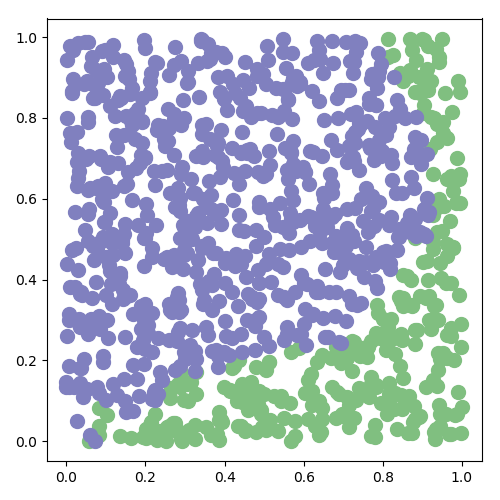} &
        \includegraphics[width=.2\linewidth]{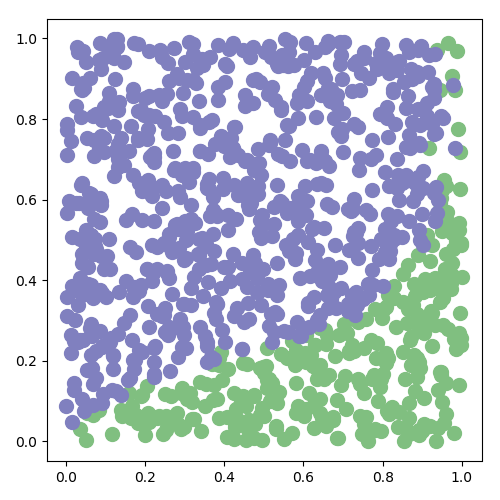} \\
        \rotatebox{90}{\qquad$k=10000$} &
        \includegraphics[width=.2\linewidth]{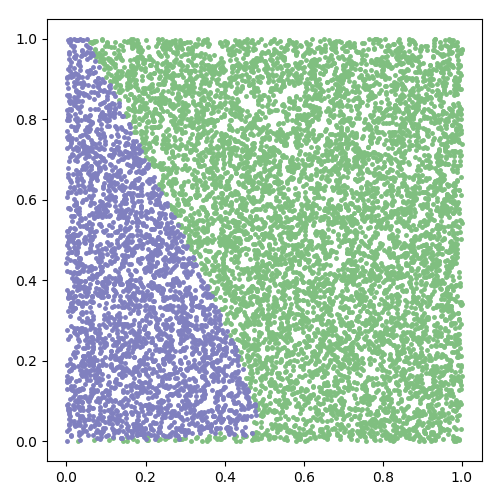} &
        \includegraphics[width=.2\linewidth]{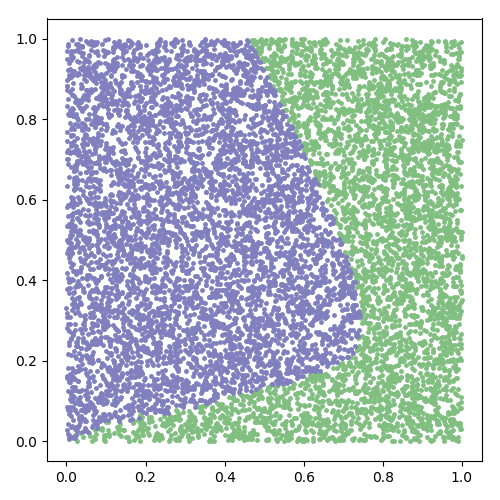} &
        \includegraphics[width=.2\linewidth]{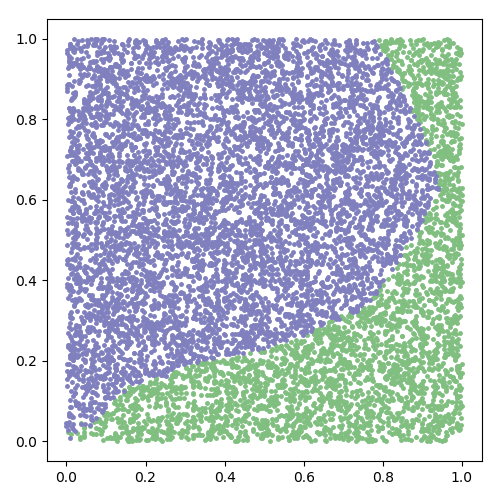} &
        \includegraphics[width=.2
        \linewidth]{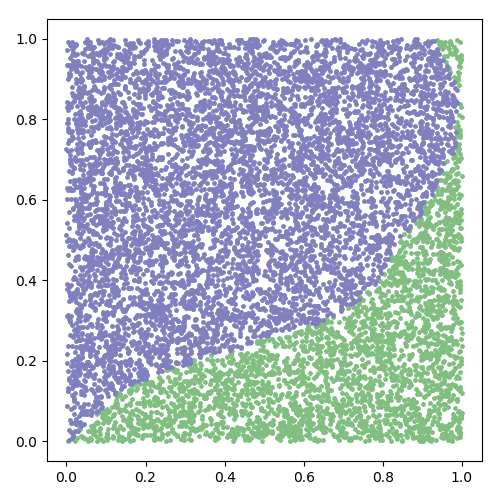} \\
    \end{tabular}
    \par\textbf{(b)} Pure equilibria for random discretizations of the game from \cref{fig:example-3}.
    \caption{Equilibria with $\ell=1$ for $\mathcal G_{\mu_k} = \langle n, \mathcal{\tilde S}_{\mu_k}, \tilde U\rangle$ computed using \cref{algo:strongalpha}, where $\mu_k$ is a distribution with a finite support of size $k$ approximating $\mu$. When $k\rightarrow+\infty$, the distributional strategies weakly converge towards the equilibria of $\mathcal G_\mu$ given in \cref{fig:example-2,fig:example-3}.}
    \label{fig:simulations-1}
\end{myfigure}

%%%%%%%%%%%%%%%%%%%%%%%%%%%%%%%%%%%%%%%%%%%%%%%%%%%%%%%%%%%%%%%%%%%%%%%%%
%%%%%%%%%%%%%%%%%%%%%%%%%%%%%%%%%%%%%%%%%%%%%%%%%%%%%%%%%%%%%%%%%%%%%%%%%
\begin{myfigure}[p!]
    \begin{tabular}{cc}
    \includegraphics[width=.45\linewidth]{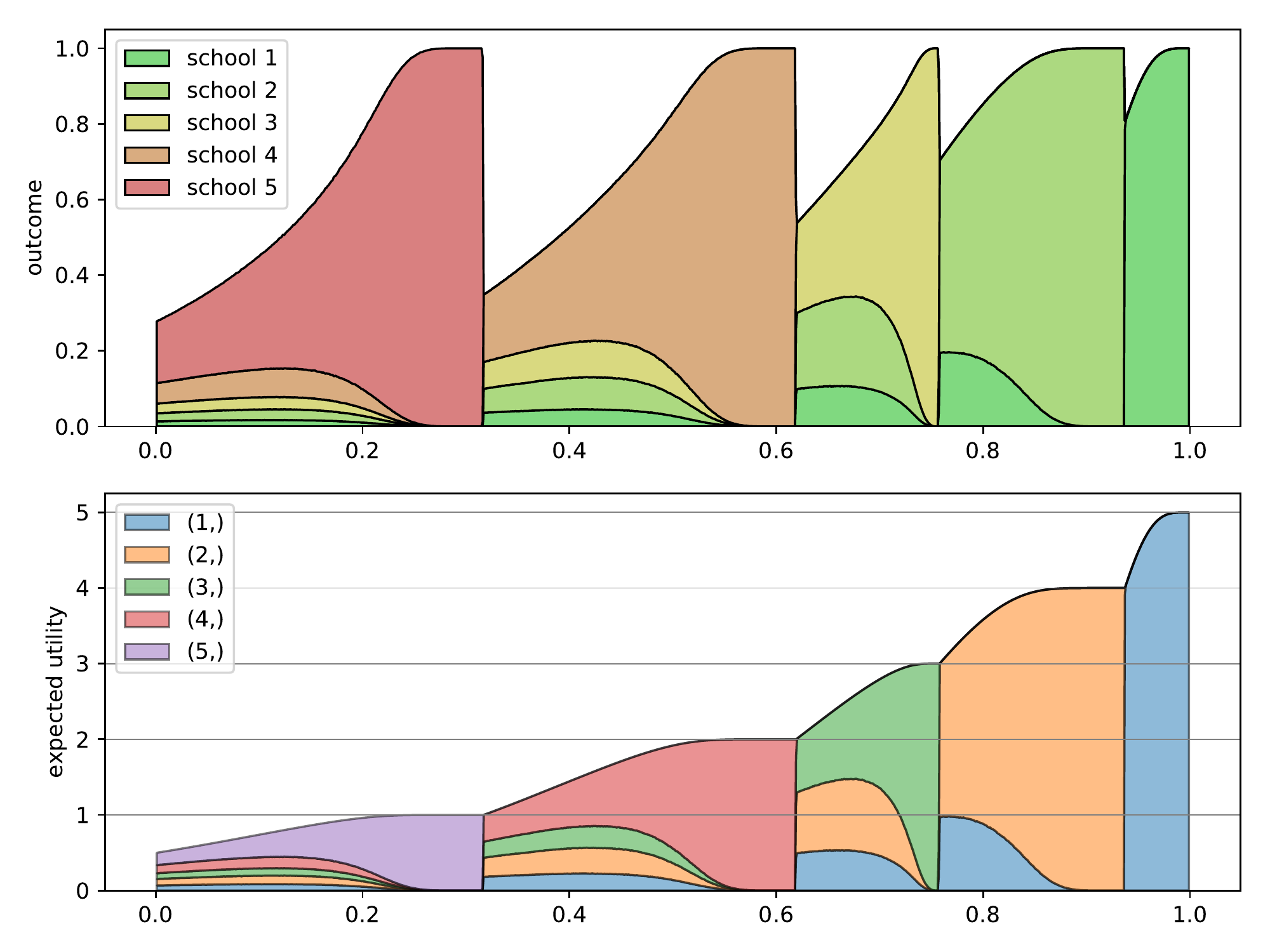} &
    \includegraphics[width=.45\linewidth]{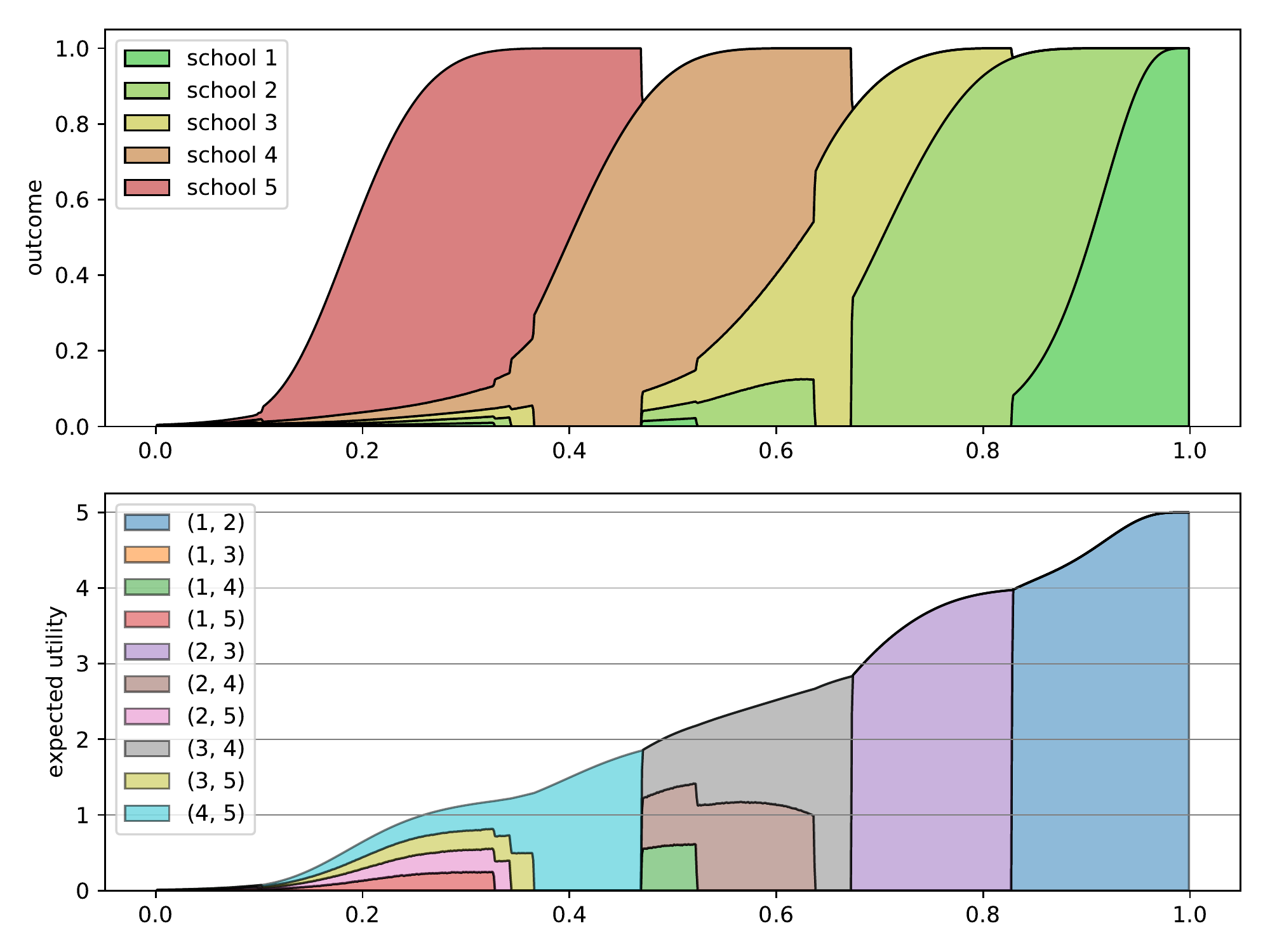} \\
    \textbf{(a)} Equilibrium with $\ell=1$. &
    \textbf{(b)} Equilibrium with $\ell=2$. \\\\
    \includegraphics[width=.45\linewidth]{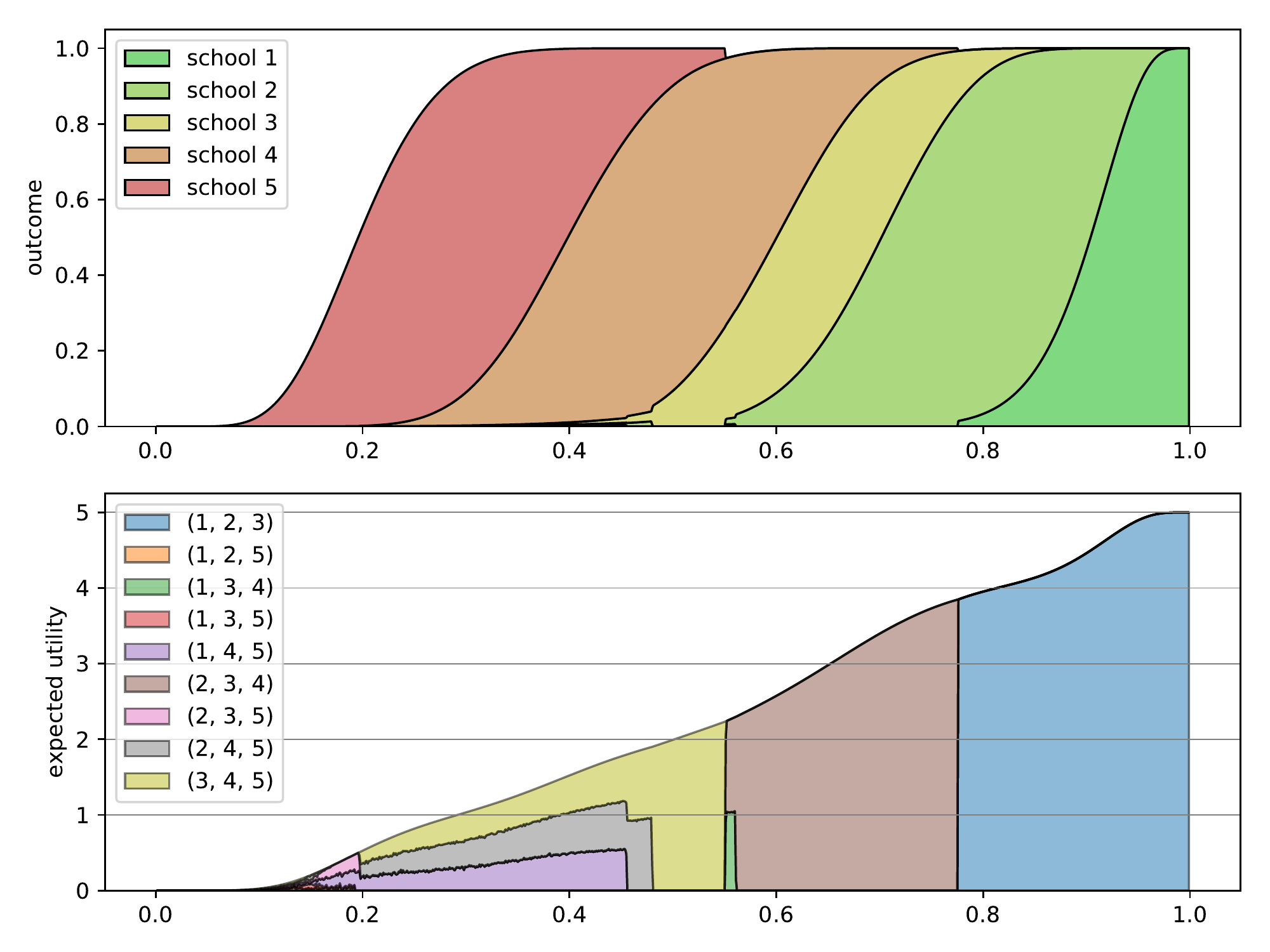} &
    \includegraphics[width=.45\linewidth]{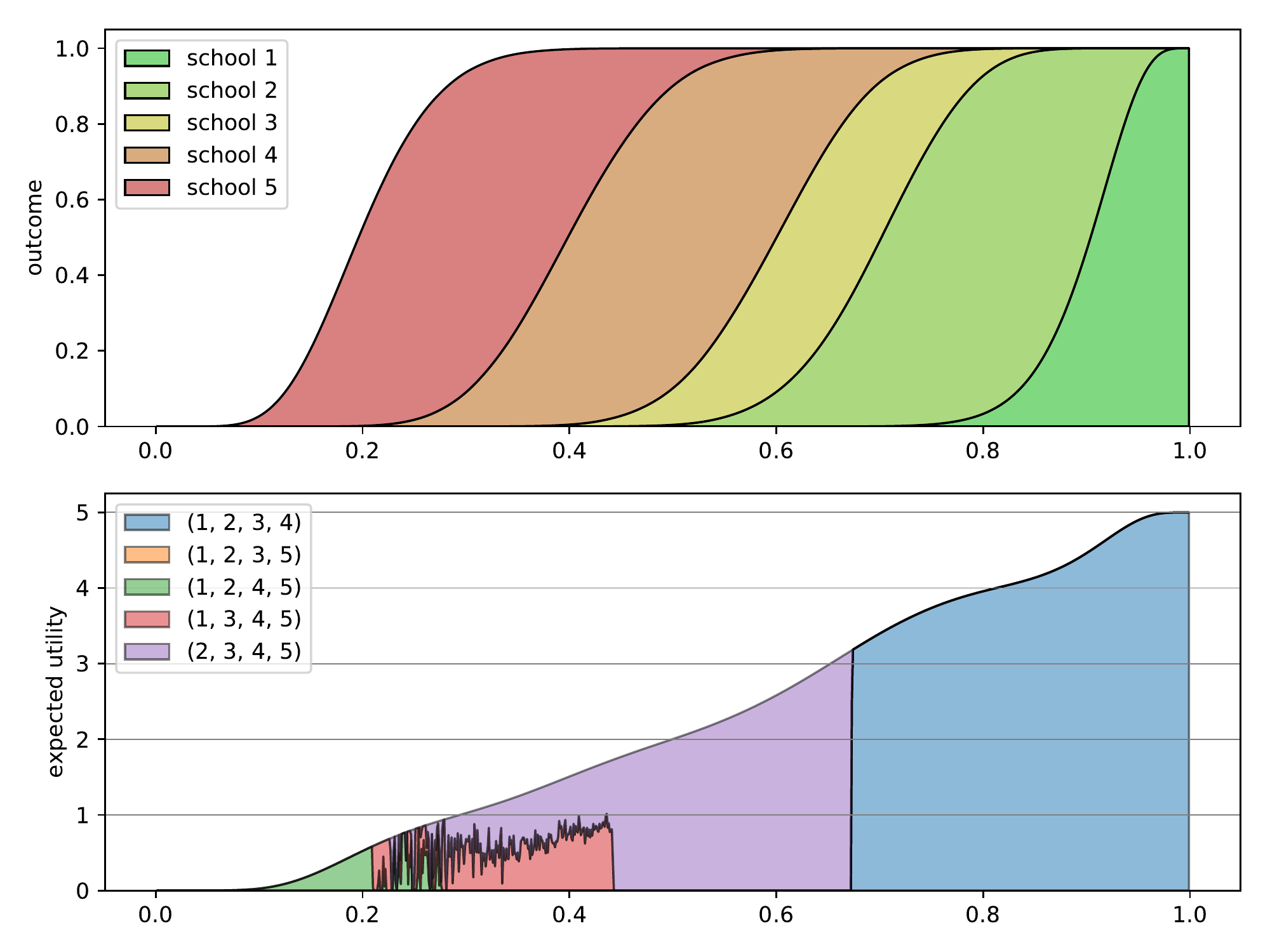} \\
    \textbf{(c)} Equilibrium with $\ell=3$. &
    \textbf{(d)} Equilibrium with $\ell=4$. \\\\
    \includegraphics[width=.45\linewidth]{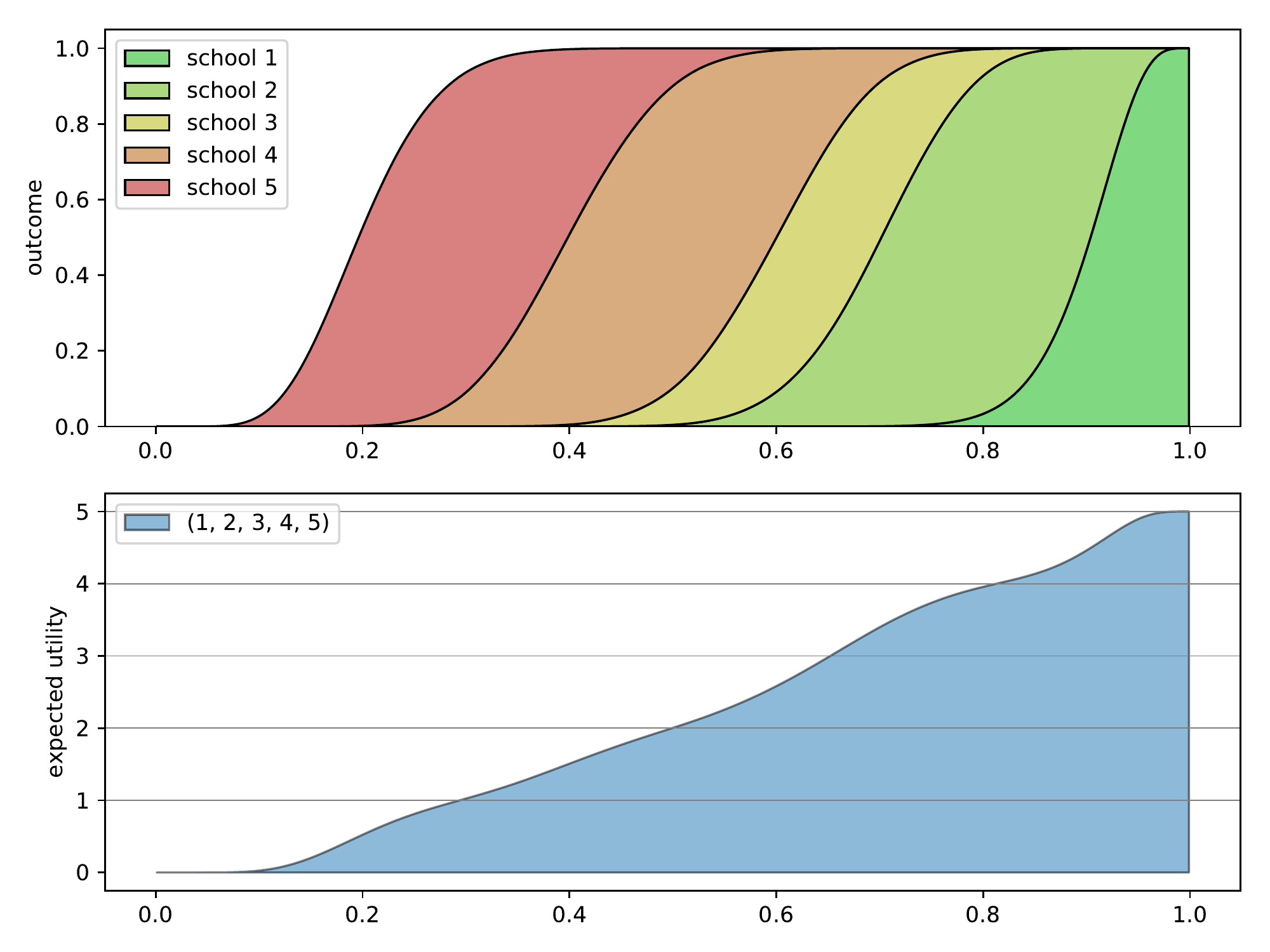} &
    %\begin{minipage}[b]{.45\linewidth}
    %Given a sequence $(t_i)_{i\geq 1}$, let $\mu_k$ be the uniform distribution over $\{t_i\}_{i\in[k]}$, and let $\tilde p_k$ be a pure equilibrium of $\mathcal G_{\mu_k}$, computed using Algorithm~\ref{algo:identical}.
    %When $k\rightarrow +\infty$, the sequence $(\tilde p_k)_{k\geq n}$ weakly converge towards a mixed equilibrium.
    %\smallbreak\noindent
    %To compute the limit, we perform the following post-processing step: for every $t \in T$, we consider the $\sqrt{k}$ types from the support of $\mu_k$ that are closest to $t$, and let $p_k(t)$ be the average strategy over those points.
    %\bigbreak
    %\end{minipage}
    \\
    \textbf{(e)} Equilibrium with $\ell=5$. &\\
    \end{tabular}
    \caption{Equilibria of the game defined in \cref{fig:example-5}, computed using \cref{algo:identical}. The top panel of each sub-figure plots the probability of each outcome, conditioned on the type of the student. In a mixed strategy, every action gives the same payoff. The bottom panel of each sub-figure decompose this payoff into the different strategy from the mixed equilibrium, conditioned on the type of the student. Exactly as we argued in \cref{fig:example-5}, the equilibrium with 3 applications per student yields almost the same outcome (in terms of expected payoff) as the one with 5 applications per student.}
    \label{fig:simulations-2}
\end{myfigure}

%%%%%%%%%%%%%%%%%%%%%%%%%%%%%%%%%%%%%%%%%%%%%%%%%%%%%%%%%%%%%%%%%%%%%%%%%
%%%%%%%%%%%%%%%%%%%%%%%%%%%%%%%%%%%%%%%%%%%%%%%%%%%%%%%%%%%%%%%%%%%%%%%%%

\section{Simulations}
\label{sec:constrained:simul}

Implementations are available at the following address:
\begin{center}
\url{https://github.com/simon-mauras/stable-matchings/tree/master/Equilibrium}
\end{center}

\paragraph{Convergence theorem with mixed equilibrium.} Both \cref{algo:strongalpha,algo:identical} compute a pure equilibrium $\tilde p_k$ of the game $\mathcal G_{\mu_k}$, in the sense that the behavioral strategy $p_k$ is pure. Using \cref{thm:convergence} we show that $\tilde p_k$ weakly converge towards an equilibrium $\tilde p$ of $\mathcal G_\mu$. When $p$ is a pure strategy, we can easily approximate $p$ by the strategy $p_k$ with a large $k$ (see \cref{fig:simulations-1}). However, if $p$ is a mixed strategy, we need an extra step to compute the limit: for every $t\in T$ we consider the $\sqrt{k}$ types from the support of $\mu_k$ that are closest to $t$, and let $p_k(t)$ be the average strategy over those points (see \cref{fig:simulations-2}).

\paragraph{Stronly $\alpha$-reducible preferences.} When the game is strongly $\alpha$-reducible, we can compute equilibrium with $\ell=1$ application per student. Special cases include when students have identical preferences (\cref{fig:example-3}) and when schools have identical preferences (\cref{fig:example-4}). In each case we implement \cref{algo:strongalpha} in Python (see \verb|identical-students.py| and \verb|identical-schools.py| respectively), to generate \cref{fig:simulations-1}.

\paragraph{Schools have identical preferences.} When schools have identical preferences, we can compute equilibrium for any $\ell \geq 1$. For simplicity our implementation also assumes that students have identical preferences. Because the complexity of \cref{algo:identical} is exponential in the number of students and the number of schools, an efficient implementation is preferable, which is the reason why we chose to have a Python script (\verb|identical-all.py|) interacting with a C++ solver (\verb|exact.cpp|).

The expensive part of \cref{algo:identical} is to evaluate the payoff of each action. To speed-up the computation, an improved solver (\verb|approximate.cpp|) outputs an approximate equilibrium of $\mathcal G_{\mu_k}$, which will converge towards an approximate equilibrium of $\mathcal G_\mu$.
The main idea is to replace the dynamic programming approach (where we compute $q$) by Monte Carlo simulations.
We randomly partition $\{t_1, \dots, t_k\}$ into $r = k/(n-1)$ sets of $n-1$ types, each corresponding to a ``run''. When considering the type $t_i$ with $1 \leq i \leq k$, we approximate $q$ by the empirical distribution of remaining capacities over the $r$ runs. \cref{fig:simulations-2} was obtained by setting $r=2\,000\,000$ and $n=50$, which runs in roughly 1 minute. 

\section{Conclusion and open questions}

In this paper, we generalized the game defined by Haeringer and Klijn \cite{haeringer2009constrained}, in a setting where students have incomplete information. We discussed the existence and the computability of equilibria in several setting. The following questions are left open for future work:
\begin{itemize}
    \item\textbf{Equilibria with 1 application per student.} In the complete information case, Haeringer and Klijn show that equilibria with 1 application per student correspond to stable matchings. As illustrated in Figure~\ref{fig:example-1.3}, the incomplete information game can have an infinite number of equilibrium. But does the set of equilibria has a lattice structure?
    \item\textbf{Unique equilibrium.} In Section~\ref{sec:constrained:efficient}, we compute equilibria with a finite number of types by eliminating dominant strategies. If each eliminated strategy strictly dominates other strategies, the equilibrium is unique. Using a convergence theorem, does unicity extends to the case where types are continuous?
    \item\textbf{Differential equations.} When combined with the convergence theorem, algorithms from Section~\ref{sec:constrained:efficient} can be seen as first order Euler methods, which eventually solve differential equations. Such equations might lead to more efficient algorithms, and a to a proof that the equilibrium is unique.
\end{itemize}

\bibliographystyle{alpha}
\bibliography{biblio.bib}
\end{document}